\documentclass[a4paper,onecolumn,10pt]{amsart}

\usepackage{amsmath}
\usepackage{amssymb}  
\usepackage{amsthm}
\usepackage{amsfonts}
\usepackage{graphicx}
\usepackage{cancel}
\usepackage{bbm}
\usepackage{url}
\usepackage{enumerate} 
\usepackage{ upgreek }
\usepackage{dsfont}
\usepackage{color}
\usepackage{subcaption}

\usepackage{makecell}
\usepackage{diagbox}
\usepackage{wrapfig}
\usepackage{rotating}
\usepackage{tabularx}
\usepackage{comment}

\usepackage{mathtools} 
\usepackage{bm}        

\newcommand{\rotsupset}{\rotatebox[origin=c]{90}{$\supset$}}
\newcommand{\rotsimeq}{\rotatebox[origin=c]{90}{$\simeq$}}

\usepackage{hyperref}
\hypersetup{
    colorlinks=true,
    linkcolor=black,
    citecolor=black,
    filecolor=magenta,      
    urlcolor=cyan,
}

\newcommand{\tr}{\textnormal{tr}}

\newcommand{\braket}[2]{\langle #1 , #2 \rangle}

\newcommand{\id}{\ensuremath{\mathds{1}}}

\newcommand{\cM}{\mathcal{M}}

\newcommand{\gC}{\mathsf{C}}
\newcommand{\gL}{\mathsf{L}}
\newcommand{\gK}{\mathsf{K}}

\def\beq{\begin{equation}}
\def\eeq{\end{equation}}
\def\bq{\begin{quote}}
\def\eq{\end{quote}}
\def\ben{\begin{enumerate}}
\def\een{\end{enumerate}}
\def\bit{\begin{itemize}}
\def\eit{\end{itemize}}

\def\ra{\rightarrow}

\def\lb{\left(}
\def\rb{\right)}
\def\lset{\lbrace}
\def\rset{\rbrace}

\def\r|{\right|}
\def\lbr{\left[}
\def\rbr{\right]}
\def\ident{\textnormal{id}}
\def\one{\id}

\newcommand\A{\mathbf{A}}
\newcommand\C{\mathbf{C}}

\newcommand\R{\mathbf{R}}
\newcommand\qua{\mathbf{H}}
\newcommand\oc{\mathbf{O}}
\newcommand\N{\mathbf{N}}
\newcommand\M{\mathcal{M}}

\newcommand{\ketbra}[1]{|#1\rangle\langle#1|}

\DeclareMathOperator{\Tr}{Tr}

\DeclareMathOperator{\Aut}{Aut}

\DeclareMathOperator{\inter}{int}
\DeclareMathOperator{\bounda}{bd}
\DeclareMathOperator{\inv}{inv}

\DeclareMathOperator{\PSD}{\mathsf{PSD}}

\DeclareMathOperator{\Iso}{Iso}

\newcommand{\tmin}{\otimes_{\min}}
\newcommand{\tmax}{\otimes_{\max}}
\newcommand{\st}{~:~}

\DeclareMathOperator{\conv}{\mathrm{conv}}

\renewcommand{\geq}{\geqslant}
\renewcommand{\leq}{\leqslant}

\newcommand{\e}{\varepsilon}

\theoremstyle{plain}
\newtheorem{thm}{Theorem}[section]
\newtheorem{lem}[thm]{Lemma}
\newtheorem{cor}[thm]{Corollary}
\newtheorem{prop}[thm]{Proposition}
\newtheorem{defn}[thm]{Definition}

\newtheorem{que}[thm]{Question}
\theoremstyle{definition}
\newtheorem{example}{Example}

\begin{document}

\title{{Annihilating Entanglement Between Cones}}

\author{Guillaume Aubrun}
\address{\small{Institut Camille Jordan, Universit\'{e} Claude Bernard Lyon 1, 43 boulevard du 11 novembre 1918, 69622 Villeurbanne cedex, France}}
\email{aubrun@math.univ-lyon1.fr}
\author{Alexander M\"uller-Hermes}
\address{\small{Institut Camille Jordan, Universit\'{e} Claude Bernard Lyon 1, 43 boulevard du 11 novembre 1918, 69622 Villeurbanne cedex, France}}
\address{\small{Department of Mathematics, University of Oslo, P.O. box 1053, Blindern, 0316 Oslo, Norway}}
\email{muellerh@math.uio.no, muellerh@posteo.net}

\maketitle
\date{\today}

\begin{abstract}
\vspace{-1cm} Every multipartite entangled quantum state becomes fully separable after an entanglement breaking quantum channel acted locally on each of its subsystems. Whether there are other quantum channels with this property has been an open problem with important implications for entanglement theory (e.g., for the distillation problem and the PPT squared conjecture). 

We cast this problem in the general setting of proper convex cones in finite-dimensional vector spaces. The entanglement annihilating maps transform the $k$-fold maximal tensor product of a cone $\gC_1$ into the $k$-fold minimal tensor product of a cone $\gC_2$, and the pair $(\gC_1,\gC_2)$ is called resilient if all entanglement annihilating maps are entanglement breaking. Our main result is that $(\gC_1,\gC_2)$ is resilient if either $\gC_1$ or $\gC_2$ is a Lorentz cone. Our proof exploits the symmetries of the Lorentz cones and applies two constructions resembling protocols for entanglement distillation: As a warm-up, we use the multiplication tensors of real composition algebras to construct a finite family of generalized distillation protocols for Lorentz cones, containing the distillation protocol for entangled qubit states by Bennett et al.~\cite{bennett1996purification} as a special case. Then, we construct an infinite family of protocols using solutions to the Hurwitz matrix equations. After proving these results, we focus on maps between cones of positive semidefinite matrices, where we derive necessary conditions for entanglement annihilation similar to the reduction criterion in entanglement distillation. Finally, we apply results from the theory of Banach space tensor norms to show that the Lorentz cones are the only cones with a symmetric base for which a certain stronger version of the resilience property is satisfied.       
\end{abstract}

\tableofcontents

\vspace{-1cm}
\section{Introduction}

Let $\M_d$ denote the set of complex $d\times d$ matrices and let $\PSD\lb \C^d\rb\subset \M_d$ denote the cone of positive semidefinite matrices with complex entries. There are two natural tensor products in the category of cones that can be specialized to the cone $\PSD(\C^d)$: The $k$-fold \emph{minimal tensor product} is given by 
\[
\PSD(\C^{d})^{\otimes_{\min} k} = \conv\lset x_1\otimes \cdots\otimes x_k~:~x_1,\ldots, x_k\in\PSD(\C^{d})\rset ,
\] 
and it is usually referred to as the set of (unnormalized) fully separable states. The $k$-fold \emph{maximal tensor product} is given by 
\[
\PSD(\C^{d})^{\otimes_{\max} k} = \lb \PSD(\C^{d})^{\otimes_{\min} k}\rb^*,
\]
with duality with respect to the Hilbert-Schmidt inner product $\braket{x}{y} = \Tr\lbr xy\rbr$ on the space of self-adjoint matrices. The maximal tensor product contains multipartite entanglement witnesses and is usually called the set of block-positive tensors. The following classes of linear maps will be central for our work: We call a linear map $P:\M_{d_A}\ra \M_{d_B}$ 
\begin{itemize}
\item \emph{entanglement breaking} if
\[
(\ident_n\otimes P)\lb \PSD\lb \C^{n}\otimes \C^{d_A}\rb\rb\subseteq \PSD(\C^{n})\otimes_{\min}\PSD(\C^{d_B}),
\]
for any $n\in\N$.
\item $k$-\emph{entanglement annihilating} if
\[
P^{\otimes k}\lb \PSD(\C^{d})^{\otimes_{\max} k}\rb\subseteq \PSD\lb\C^{d}\rb^{\otimes_{\min} k}.
\]
\item \emph{entanglement annihilating} if it is $k$-entanglement annihilating for all $k\in\N$.
\end{itemize} 
Every entanglement breaking map $P:\M_{d_A}\ra \M_{d_B}$ admits a decomposition $P(\cdot)=\sum^N_{i=1} y_i\Tr\lbr x_i \cdot\rbr$ with $x_i\in \PSD\lb \C^{d_A}\rb$ and $y_i\in \PSD\lb \C^{d_B}\rb$ (see~\cite{horodecki2003entanglement}). From this, it is easy to see that any entanglement breaking map is entanglement annihilating. However, it is unknown whether the converse holds as well:

\begin{que}\label{que:Main}
Are entanglement annihilating maps always entanglement breaking?
\end{que}
In this article, we study Question \ref{que:Main} in the general setting of convex cones. We identify an infinite family of cones where its answer is `Yes', and discuss implications for potential entanglement annihilating maps on $\PSD\lb \C^d\rb$. Finally, we identify a candidate where the analogue of Question \ref{que:Main} might have a negative answer, and we discuss implications from the theory of Banach space tensor norms.

\subsection{Motivation and history} 

For any $d\in \N$, let $\vartheta_d:\M_d\ra \M_d$ denote the transpose map in the computational basis. The distillation problem~\cite{horodecki1998mixed,divincenzo2000evidence} asks whether every quantum state $\rho\in \PSD(\C^{d_A}\otimes \C^{d_B})$ with non-positive partial transpose (NPPT), i.e., such that $(\ident_{d_A}\otimes \vartheta_{d_B})\lb \rho\rb\notin \PSD(\C^{d_A}\otimes \C^{d_B})$, can be transformed into the maximally entangled state
\[
\omega_2 = \frac{1}{2}\begin{pmatrix} 1 & 0 & 0 & 1 \\ 0 & 0 & 0 & 0 \\ 0 & 0 & 0 & 0 \\ 1 & 0 & 0 & 1\end{pmatrix}\in \PSD\lb \C^2\otimes \C^2\rb,
\]
by taking tensor powers $\rho\mapsto \rho^{\otimes k}$ and applying local quantum operations and classical communication (see~\cite{chitambar2014everything} for the precise definition of this class of linear maps). Entangled quantum states for which such a transformation is not possible are called \emph{bound entangled}. Recall that a linear map $P:\M_{d_A}\ra \M_{d_B}$ is called $2$-positive if $\ident_2\otimes P$ is positive, and it is called completely positive if $\ident_n\otimes P$ positive for every $n\in\N$. The distillation problem has been shown~\footnote{This was essentially shown in~\cite{divincenzo2000evidence} with a minor modification using the twirling techniques from~\cite{muller2016positivity,muller2018decomposability}.} to be equivalent to the following elementary question: Are all linear maps $P:\M_{d_A}\ra \M_{d_B}$ for which $P^{\otimes k}$ is $2$-positive for every $k\in\N$ necessarily completely positive? Recently, the second author established a one-way implication of a similar form: 

\begin{thm}[Theorem 4 in~\cite{muller2016positivity}]\label{thm:tsPosImpliesNPPTBE} 
The existence of bound entangled quantum states with non-positive partial transpose would follow from the existence of a positive map $P:\M_{d_A}\ra \M_{d_B}$ satisfying the following two conditions:
\begin{enumerate}
\item Neither $P$ nor $\vartheta_{d_B}\circ P$ are completely positive.
\item For every $k\in\N$ the map $P^{\otimes k}$ is positive.
\end{enumerate}
\end{thm}

Linear maps $P$ for which $P^{\otimes k}$ is positive for every $k\in\N$ are called \emph{tensor-stable positive}. In~\cite[Theorem 5]{muller2016positivity} it has been shown that an entanglement annihilating map $T:\M_{d_A}\ra \M_{d_B}$ which is not entanglement breaking could be used to construct a tensor-stable positive map $P:\M_{d^2_A}\ra \M_{d^2_B}$ such that neither $P$ iself nor $\vartheta_{d^2_B}\circ P$ are completely positive. By Theorem \ref{thm:tsPosImpliesNPPTBE} we have the following: 
 
\begin{thm}[Entanglement annihilation implies NPPT bound entanglement]\label{thm:ResImplNPPTBE}
The existence of bound entangled quantum states with non-positive partial transpose would follow from the existence of an entanglement annihilating map that is not entanglement breaking.
\end{thm}

It should be noted that the existence of entanglement annihilating maps as in Theorem \ref{thm:ResImplNPPTBE} would have other important consequences in entanglement theory. For example it would also provide a counterexample to the so-called PPT squared conjecture~\cite{christandl2012PPT,christandl2019composed} by combining \cite[Theorem 5]{muller2016positivity} with \cite[Theorem 6.1]{muller2018decomposability}. From a more practical viewpoint it is important to study entanglement annihilation in order to gain understanding in the limitations of quantum technologies. While entanglement annihilating physical processes might be much less noisy than entanglement breaking processes, they will still destroy any entanglement in a quantum-many-body system when affecting each local site. 

Finally, we should emphasize that in this article we use a slightly different terminology than previous literature on entanglement annihilation, e.g.,~\cite{moravvcikova2010entanglement,filippov2012local,filippov2013bipartite,filippov2013dissociation}. Previously, a quantum channel $T:\M_{d_A}\ra \M_{d_B}$ was called $k$-(locally) entanglement annihilating if
\[
T^{\otimes k}\lb \PSD\lb (\C^{d_A})^{\otimes k}\rb\rb \subseteq \PSD(\C^{d_B})^{\otimes_{\min} k} .
\]
Question \ref{que:Main} was asked in this context in~\cite{moravvcikova2010entanglement}. We have changed the terminology to obtain a more general class of maps that can be defined for any pair of cones (see next section). Our definition is more restrictive since even non-physical forms of entanglement (cf.,~\cite{ALPP21}) are required to be annihilated, and it is the strongest form of entanglement annihilation in any generalized probabilistic theory (see~\cite{lami2018non} for an introduction to GPTs). Still, the answer of Question \ref{que:Main} is elusive even for this stronger notion. 

It should be noted that for every finite $k\in\N$ there are $k$-entanglement annihilating maps that are not entanglement breaking:

\begin{prop}\label{prop:kEA}
For every $k\in \N$ and any $d\geq 2$, there exists a linear map $P:\M_d\ra \M_d$ such that:
\begin{enumerate}
\item $P$ is not entanglement breaking.
\item $P$ is $k$-entanglement annihilating. 
\end{enumerate}
\end{prop}

In Appendix~\ref{app:ProofOfPropkEA}, we will prove a general statement from which Proposition \ref{prop:kEA} will follow. Proposition \ref{prop:kEA} even shows the existence of $k$-entanglement breaking maps $P:\M_2\ra \M_2$ that are not entanglement breaking, but we know that such maps cannot be $k$-entanglement annihilating for every $k\in\N$. Therefore, Proposition \ref{prop:kEA} should \emph{not} be seen as evidence for a negative answer of Question \ref{que:Main}. 

In the special case of $k=2$, we also want to point out the following proposition (for a proof see Appendix~\ref{app:ProofOfProp}) that can be seen as an improved version of the PPT squared conjecture in dimension $d=3$ proved in~\cite{christandl2019composed,chen2019positive}:
\begin{prop}\label{prop:StrongPPT2}
Consider linear maps $T,S:\M_3\ra \M_3$. If the maps $T,S,\vartheta_3\circ T$ and $\vartheta_3\circ S$ are completely positive, then 
\[
(T\otimes S)\lb \PSD(\C^3)\otimes_{\max} \PSD(\C^3)\rb \subseteq \PSD(\C^3)\otimes_{\min} \PSD(\C^3) .
\]
\end{prop}

The previous proposition gives many examples of $2$-entanglement annihilating maps that are not entanglement breaking (e.g., by using the maps corresponding to the states constructed and referenced in~\cite{clarisse2006construction}). We will now recast Question \ref{que:Main} in the general setting of convex cones in finite-dimensional vector spaces.

\subsection{Entanglement annihilation on proper cones} 

A convex cone $\gC \subset V$ in a finite-dimensional real vector space $V$ is called \emph{proper} if it is closed and satisfies the relations $\gC-\gC=V$ and $\gC \cap (-\gC) = \{0\}$. When $\gC \subset V$ is a proper cone, we define its \emph{dual cone} as 
\[ \gC^* = \{ \phi \in V^* \st \phi(x) \geq 0 \textnormal{ for every } x \in \gC \}. \]
The dual cone $\gC^*$ is a proper cone in $V^*$ and when we identify the bidual $V^{**}$ with $V$, the relation $\gC^{**}=\gC$ holds\footnote{This fact is sometimes called the \emph{bipolar theorem}.}. Let $V_1,V_2$ denote finite-dimensional vector spaces and $\gC_1\subset V_1$ and $\gC_2\subset V_2$ proper cones. We define the minimal tensor product of $\gC_1$ and $\gC_2$ as 
\[
\gC_1\otimes_{\min} \gC_2 = \text{conv}\lset x\otimes y~:~x\in \gC_1, y\in \gC_2 \rset \subset V_1\otimes V_2,
\] 
and the maximal tensor product as  
\[
\gC_1\otimes_{\max} \gC_2 = (\gC_1^*\otimes_{\min} \gC_2^*)^* \subset V_1\otimes V_2 .
\]
It is easy to check that both $\gC_1\otimes_{\min} \gC_2$ and $\gC_1\otimes_{\max} \gC_2$ are proper cones whenever $\gC_1$ and $\gC_2$ are proper cones and by iterating these constructions, we also define inductively the minimal and maximal tensor powers of a proper cone $\gC$: if $k \geq 1$ is an integer, then 
\[ \gC^{\tmin (k +1)} = \gC \tmin \gC^{\tmin k} \textnormal{ and } \gC^{\tmax (k +1)} = \gC \tmax \gC^{\tmax k} \]
with the convention that $\gC ^{\tmin 1} = \gC^{\tmax 1}=\gC$. 

By analogy with the case of quantum mechanics, tensors which belong to $\gC_1\otimes_{\max} \gC_2$ but not to $\gC_1\otimes_{\min} \gC_2$ are called \emph{entangled}. The main results from \cite{ALPP21} characterizes the existence of entanglement: The equality $\gC_1\otimes_{\min} \gC_2 = \gC_1\otimes_{\max} \gC_2$ holds if and only if $\gC_1$ or $\gC_2$ is classical, i.e., it is isomorphic to the cone $\R_+^n$ for some $n\in\N$.

The following classes of linear maps can be defined naturally in the category of cones: For proper cones $\gC_1\subset V_1$ and $\gC_2\subset V_2$, a linear map $P:V_1\ra V_2$ is called
\begin{itemize}
\item \emph{$(\gC_1,\gC_2)$-entanglement breaking} if it can be written as a finite sum 
\[
P=\sum^k_{i=1}x_i \phi_i ,
\] 
for $x_1,\ldots ,x_k\in \gC_2$ and $\phi_1,\ldots ,\phi_k\in \gC^*_1$.
\item \emph{$(\gC_1,\gC_2)$-entanglement annihilating} if
\[
P^{\otimes k}\lb \gC_1^{\otimes_{\max} k}\rb\subseteq \gC_2^{\otimes_{\min} k}
\] for every $k\in \N$.
\end{itemize}
If it is clear which cones are being considered, we will omit the prefix $(\gC_1,\gC_2)$- in these definitions. Note that these classes of linear maps generalize the aforementioned classes of the cones $\PSD(\C^d)$. Motivated by Question \ref{que:Main} we define: 

\begin{defn}[Resilience]
Let $\gC_1 \subset V_1$, $\gC_2 \subset V_2$ be proper cones. The pair $(\gC_1,\gC_2)$ is called \emph{resilient} if every $(\gC_1,\gC_2)$-entanglement annihilating map $P:V_1\ra V_2$ is entanglement breaking. We say that $\gC$ is resilient if $(\gC,\gC)$ is resilient.
\end{defn}

Intuitively, a pair $(\gC_1,\gC_2)$ is resilient if for every positive map $P$ that is not entanglement breaking, there are tensors in $\gC^{\otimes_{\max} k}_1$ for some $k\in\N$ whose entanglement withstands the action of $P^{\otimes k}$. Classical cones are always resilient since any positive map is entanglement breaking. Whenever the cones $\gC_1\subset V_1$ and $\gC_2\subset V_2$ are not classical and for any fixed $k\in\N$ we show in Appendix \ref{app:ProofOfPropkEA} (see Theorem \ref{thm:ConcretekEA1}), that there are linear maps $P:V_1\ra V_2$ which are not entanglement breaking, but such that 
\[
P^{\otimes k}\lb \gC_1^{\otimes_{\max} k}\rb\subseteq \gC_2^{\otimes_{\min} k}.
\]
Therefore, the following question is non-trivial:
\begin{que}
Are there non-resilient pairs of proper cones? Equivalently, are there entanglement annihilating maps which are not entanglement breaking ?
\label{que:MainQuestion}
\end{que}

\subsection{Main results} Most previous results on entanglement annihilating maps and Question \ref{que:Main} exploit the theory of entanglement distillation, and in particular that all NPPT quantum states $\rho\in \PSD(\C^{d_A}\otimes \C^{d_B})$ with $\min(d_A,d_B)=2$ are distillable~\cite{divincenzo2000evidence,dur2000distillability}. As a consequence, the pairs $(\PSD(\C^{d_A}),\PSD(\C^{d_B}))$ are resilient when $\min(d_A,d_B)=2$ (see the proof of Lemma 3.2 in \cite{christandl2019composed}). It seems difficult to generalize this result to situations when $\min(d_A,d_B)>2$, but our approach suggests a different strategy: For $n\in\N$ consider the Lorentz cones $\gL_n\subset \R^{n+1}$ given by 
\[
\gL_n = \lset (t,x)\in \R\oplus \R^{n} ~:~\|x\|_2\leq t\rset
\]
where $\|\cdot\|_2$ is the standard Euclidean norm on $\R^n$. The cone $\PSD(\C^2)$ is isomorphic to the Lorentz cone $\gL_3$. This can be seen using the spinor representation (see~\cite[p.32]{aubrun2017alice}) or by realizing $\PSD(\C^2)$ as the cone over the Bloch ball. We have: 
\[
\begin{array}[t]{@{}c@{}}
\gL_2 \\
\rotsupset \\
\PSD\lb \C^2\rb \\
\rotsimeq \\
      \gL_3 \\
      \rotsupset \\
      \gL_4 \\
      \rotsupset \\
      \gL_5 \\
      \rotsupset \\
      \vdots 
  \end{array} 
  \begin{array}[t]{@{}c@{}}
 \\
 \\
  \subset  \\
\\
      \\
      \\
      \\
      \\
      \\
      \\
  \end{array}
\begin{array}[t]{@{}c@{}}
 \\
 \\
\PSD\lb \C^3\rb \\
\\
      \\
      \\
      \\
      \\
      \\
      \\
  \end{array}
    \begin{array}[t]{@{}c@{}}
 \\
 \\
\subset \\
\\
      \\
      \\
      \\
      \\
      \\
      \\
  \end{array}
\begin{array}[t]{@{}c@{}}
 \\
 \\
\PSD\lb \C^4\rb \\
\\
      \\
      \\
      \\
      \\
      \\
      \\
  \end{array}
    \begin{array}[t]{@{}c@{}}
 \\
 \\
\subset \\
\\
      \\
      \\
      \\
      \\
      \\
      \\
  \end{array}
\begin{array}[t]{@{}c@{}}
 \\
 \\
\PSD\lb \C^5\rb \\
\\
      \\
      \\
      \\
      \\
      \\
      \\
  \end{array}
  \begin{array}[t]{@{}c@{}}
 \\
 \\
\subset \cdots\\
\\
      \\
      \\
      \\
      \\
      \\
      \\
  \end{array}
\]
The following theorem is our main result:

\begin{thm}
The pairs $(\gL_n, \gC)$ and $(\gC, \gL_n)$ are resilient for every proper cone $\gC$ and every $n\in\N$. In particular, the Lorentz cone $\gL_n$ is resilient for every $n\in\N$.
\label{thm:Main1}
\end{thm}

Theorem \ref{thm:Main1} is a natural generalization and strengthening of the aformentioned results on resilience of $(\PSD(\C^{d_A}),\PSD(\C^{d_B}))$ when $\min(d_A,d_B)=2$. Its proof uses first the symmetries of the Lorentz cones and second a series of ``distillation protocols'' on the Lorentz cones $\gL_n$. For $n\leq 9$ we construct a family of protocols from the multiplication tensors of certain split-algebras and the normed division algebras. These protocols contain the original distillation protocol for entangled quantum states $\rho\in\PSD\lb \C^2\otimes \C^2\rb$ introduced in~\cite{bennett1996purification} as a special case. Finally, we construct a different class of protocols based on solutions of the Hurwitz matrix equations showing resilience of $\gL_n$ for every $n\geq 10$ (and also for smaller $n$).  

In the context of entanglement distillation, the reduction criterion~\cite{horodecki1999reduction} gives a sufficient condition for quantum states to be distillable. It uses the so-called reduction map $R:\M_{d}\ra \M_d$
\begin{equation}
R(X) = \Tr(X)\one_d - X ,
\label{equ:Red}
\end{equation}
and quantum states $\rho\in\PSD\lb \C^d\otimes \C^d\rb$ satisfying $(\ident_d\otimes R)(\rho)\ngeq 0$ are distillable. Mathematically, this criterion is based on the fact that the map $\vartheta_d\circ R$ factors (completely positively) through the cone $\PSD\lb \C^2\rb$, which is closely related to the so-called Schmidt number of Werner states (see~\cite{terhal2000schmidt}). Again motivated by the equivalence $\gL_3\simeq \PSD(\C^2)$ we generalize this result (in a certain sense) to maps factoring through a Lorentz cone $\gL_n$. Examples of such maps include the Breuer--Hall map and projections onto spin factors (see Section \ref{sec:PosMapsFacThroughLor} for definitions and details). Each such map gives necessary conditions for positive maps $P:\M_{d_A}\ra \M_{d_B}$ to be entanglement annihilating. 

Finally, we consider resilience in the case of cones with a symmetric base, or equivalently cones $\gC_X$ associated to a finite-dimensional normed space $X$. Using recent results obtained by the authors~\cite{paperB} on regularizations of Banach space tensor norms, we demonstrate a partial version of resilience where a restricted form of entanglement coming from the Banach space structure is annihilated by certain maps that are not entanglement breaking. This result illuminates the limitations of the methods leading to resilience of Lorentz cones, and we show that they cannot show resilience of any other cone $\gC_X$ with symmetric base. Finally, we discuss the cone over the finite-dimensional $\ell_1$-spaces, which is a candidate for a non-resilient cone. 

Our article is structured as follows: 
\begin{itemize}
\item In Section \ref{sec:PrelAndNot} we review some preliminaries and notation. 
\item In Section \ref{sec:Symmetrization} we show how to exploit symmetries in order to simplify the study of resilience of cones. Specifically, we will develop techniques for cones with a symmetric base (Section \ref{sec:ConeWithSymBaseRed}) and for cones with enough symmetries (Section \ref{sec:EnoughSymm}).
\item In Section \ref{sec:lorentz} we prove Theorem \ref{thm:Main1} on the resilience of Lorentz cones.
\item In Section \ref{sec:fact} we study positive maps factoring through cones and how they can be used to study resilience of cones. In Section \ref{sec:GenFactTheory} we develop the general theory and in Section \ref{sec:PosMapsFacThroughLor} we specialize to positive maps between cones of positive semidefinite matrices factoring through Lorentz cones giving rise to generalized reduction criteria.  
\item In Section \ref{sec:ConeSymBase} we study cones with symmetric base and connections between resilience and the theory of Banach space tensor norms. We first review the results from~\cite{paperB} in Section \ref{sec:SummaryPaperB}, and then we study the annihilation of particular forms of entanglement related to Banach space tensor norms in Section \ref{sec:TensorProdFromNorms}. Finally, we discuss a potential candidate for a resilient cone in Section \ref{sec:CandidateForResilience}. 
\end{itemize}

\section{Preliminaries and notation}\label{sec:PrelAndNot}

Unless explicitly stated, all the vector spaces we consider are assumed to be finite-dimensional vector spaces over $\R$.

\subsection{Classes of linear maps and their correspondence to tensor products} \label{subsec:tensor-products}

Our main object of study are linear maps between vector spaces $V_1$, $V_2$. Consider two proper cones $\gC_1 \subset V_1$ and $\gC_2 \subset V_2$. A linear map $P : V_1 \to V_2$ is said to be $(\gC_1,\gC_2)$-\emph{positive} if it satisfies the relation $P(\gC_1) \subseteq \gC_2$. The class of $(\gC_1,\gC_2)$-positive maps forms itself a proper cone which we denote by  $\mathcal{P}(\gC_1,\gC_2)$. 

It is natural to identify a linear map $P:V_1\ra V_2$ with the tensor $\hat{P}\in V^*_1\otimes V_2$ satisfying the relation
\[
\phi\lb P(x)\rb = \lb x \otimes \phi\rb(\hat{P})
\]
for every $\phi\in V^*_2$ and every $x\in V_1$. In the right-hand side of that formula, $x$ is considered as an element of $V_1^{**}$.  This correspondence $P\leftrightarrow \hat{P}$ defines an isomorphism between linear maps $P:V_1\ra V_2$ and tensors $\hat{P}\in V^*_1\otimes V_2$. When $V_1=\M_{d_1}$ and $V_2=\M_{d_2}$ this isomorphism is sometimes called the Choi--Jamiolkowski isomorphism~\cite{choi1975completely,jamiolkowski1972linear}.  

\begin{lem} \label{lem:maps-vs-tensors}
Let $V_1$, $V_2$ be vector spaces, $\gC_1 \subset V_1$, $\gC_2 \subset V_2$ be proper cones and $P: V_1 \to V_2$ be a linear map. Then
\begin{enumerate}
    \item the map $P$ is $(\gC_1,\gC_2)$-positive if and only if $\hat{P} \in \gC_1^* \tmax \gC_2$,
    \item the map $P$ is $(\gC_1,\gC_2)$-entanglement breaking if and only if $\hat{P} \in \gC_1^* \tmin \gC_2$.
\end{enumerate}
\end{lem}

\begin{proof}
Since $\gC_2=\gC_2^{**}$, the positivity of $P$ is equivalent to the fact that $\phi(P(x)) \geq 0$ for every $x \in \gC_1$ and $\phi \in \gC_2^*$. Using the definition of the maximal tensor product, this is equivalent to the condition $\hat{P} \in \gC_1^* \tmax \gC_2$. The second statement is an easy consequence of the definitions.
\end{proof}

We now characterize entanglement breaking maps as the maps which destroy entanglement when applied to one part of a tensor product. This statement extends a well known fact in quantum information theory (\cite{horodecki2003entanglement}):

\begin{prop} \label{prop:EB}
Let $V_1$, $V_2$ be vector spaces, $\gC_1 \subset V_1$, $\gC_2 \subset V_2$ be proper cones and $P: V_1 \to V_2$ be a linear map. The following are equivalent
\begin{enumerate}
    \item the map $P$ is $(\gC_1,\gC_2)$-entanglement breaking,
    \item for every proper cone $\gC$ in a vector space $V$, we have
    \[ (\ident_V \otimes P) \left( \gC \tmax \gC_1 \right) \subseteq \gC \tmin \gC_2 .\]
\end{enumerate}
\end{prop}

\begin{proof}
Assume (1), so that $P$ can be decomposed as $P(\cdot)=\sum \phi_i(\cdot)x_i$
for some $x_i \in \gC_2$ and $\phi_i \in \gC^*_1$. Consider a tensor $z \in \gC \tmax \gC_1$. It follows from the definition of the maximal tensor product that $(\ident_V \otimes \phi_i)(z) \in \gC$ for every $i$, and therefore $(\ident_V \otimes P)(z) = \sum (\ident_V \otimes \phi_i)(z) \otimes x_i$ belongs to $\gC \tmin \gC_2$. This shows (2).

Conversely, assume (2) and choose $V=V_1^*$, $\gC = \gC_1^*$. By Lemma \ref{lem:maps-vs-tensors}, the tensor $\widehat{\ident_{V_1}}$ belongs to $\gC_1^* \tmax \gC_1$. Since $\hat{P} =  (\ident_{V_1^*} \otimes P)(\widehat{\ident_{V_1}})$ belongs to $\gC \tmin \gC_2$, using again Lemma \ref{lem:maps-vs-tensors} shows that $P$ is entanglement breaking.
\end{proof}

\subsection{Duality between classes of maps} 

Consider vector spaces $V_1$, $V_2$ and a linear map $P : V_1 \to V_2$. We denote by $P^* : V_2^* \to V_1^*$ its adjoint. The following lemma is easy to check.

\begin{lem} \label{lem:dual}
Let $\gC_1 \subset V_1$, $\gC_2 \subset V_2$ be proper cones and $P : V_1 \to V_2$ a linear map. Then
\begin{enumerate}
    \item $P$ is $(\gC_1,\gC_2)$-positive if and only if the map $P^*$ is $(\gC_2^*,\gC_1^*)$-positive.
    \item $P$ is $(\gC_1,\gC_2)$-entanglement breaking if and only if the map $P^*$ is $(\gC_2^*,\gC_1^*)$-entanglement breaking.
    \item $P$ is $(\gC_1,\gC_2)$-entanglement annihilating if and only if the map $P^*$ is $(\gC_2^*,\gC_1^*)$-entanglement annihilating.
    \item The pair $(\gC_1,\gC_2)$ is resilient if and only if the pair $(\gC_2^*,\gC_1^*)$ is resilient.
\end{enumerate}
\end{lem}

Another useful lemma is obtained by considering trace duality.

\begin{lem} \label{lem:dualityEB}
Let $\gC_1 \subset V_1$ and $\gC_2 \subset V_2$ be proper cones, and $P: V_1 \to V_2$ a linear map. The following are equivalent
\begin{enumerate}
    \item $P$ is $(\gC_1,\gC_2)$-entanglement breaking,
    \item for every positive map $Q \in \mathcal{P}(\gC_2,\gC_1)$, we have $\Tr[Q \circ P] \geq 0$.
\end{enumerate}
\end{lem}

\begin{proof}
Let $\iota: V_2^* \otimes V_1 \to (V_1^* \otimes V_2)^*$ be the canonical identification. If $Q : V_2 \to V_1$ is a linear map, then we have
\[ (\iota(\hat{Q}))(\hat{P}) = \Tr[Q \circ P] ,\]
which is easy to check when $P$ and $Q$ have rank $1$ and the general case follows from linearity. Lemma \ref{lem:maps-vs-tensors} shows that condition (1) is equivalent to $\hat{P} \in \gC_1^* \tmin \gC_2$, and condition (2) is equivalent to $\hat{P} \in (\iota(\gC_2^* \tmax \gC_1))^*$. The result follows since
\[ \iota(\gC_2^* \tmax \gC_1) = (\gC_1^* \tmin \gC_2)^*. \qedhere \]
\end{proof}

\subsection{Proper cones associated to convex bodies and normed spaces}
\label{sec:properConesNormedSpaces}

Let $B \subset V$ be a convex body, i.e., a compact convex set with non-empty interior. We define the \emph{cone over $B$} as
\[\gC_B = \lset (t,x)\in \R \oplus V ~:~ t\geq 0, \ x\in tB\rset ,\]
which is a proper cone in $\R \oplus V$. It is an elementary fact that any proper cone is isomorphic to the cone over some convex body. We will sometimes consider cones $\gC_B$ over a symmetric convex body $B$, i.e., such that $-B=B$. Since symmetric convex bodies are unit balls of norms and vice-versa, we may equivalently consider cones of the form
\[
\gC_X = \lset (t,x)\in \R\oplus X ~:~t\geq \|x\|_X\rset,
\]
for a normed space $X$. Important examples arise from the $\ell_p$-spaces $\ell^n_p = \lb\R^n,\|\cdot\|_p\rb$ and we note that the Lorentz cones are given by $\gL_n=\gC_{\ell^n_2}$. 

It is not surprising that $(\gC_X,\gC_Y)$-positivity of certain maps from $\R\oplus X$ to $\R\oplus Y$ can be characterized using the normed spaces $X$ and $Y$. We say that a linear map $Q : \R \oplus X \to \R \oplus Y$ is a \emph{central map} if it has the form
\[ Q = \alpha \oplus P : (t,x) \mapsto (\alpha t, P(x)) \]
for $\alpha \in \R$ and $P :X \to Y$ a linear map. It is easy to verify that the linear map $\alpha \oplus P$ is
\begin{itemize}
\item $(\gC_X,\gC_Y)$-positive if and only if $\|P\|_{X\ra Y}\leq \alpha$.
\item $(\gC_X,\gC_Y)$-entanglement breaking if and only if $\|P\|_{N(X\ra Y)}\leq \alpha$.
\end{itemize}
Here, we used the \emph{nuclear norm} given by 
\[ 
\|P\|_N = \|P\|_{N(X \to Y)} = \inf\, \sum^n_{i=1} \|y_i\|_Y\|x^*_i\|_{X^*} ,
\]
where the infimum is over $n\in\N$ and decompositions
\[
P=\sum^n_{i=1} y_{i}x^*_i , \text{ with } y_1,\ldots ,y_n\in Y \text{ and } x^*_1,\ldots , x^*_n\in X^* .
\] 

\subsection{Basic properties of resilient cones and entanglement annihilation} \label{sec:coding}

In this section, we establish some basic properties of entanglement annihilating maps and resilient cones. We start by characterizing entanglement annihilating maps as maps which stay positive under certain transformations resembling the encoding/decoding operations (specifically, separable operations~\cite{chitambar2014everything}) from quantum information theory. 

Consider a proper cone $\gC$ inside a vector space $V$ and an integer $k \geq 1$. We say that a linear map $E:V \to V^{\otimes k}$ is a $\gC$-\emph{encoder} if it is $(\gC,\gC^{\tmax k})$-positive, and that a linear map $D :V^{\otimes k} \to V$ is a $\gC$-\emph{decoder} if it is $(\gC^{\tmin k},\gC)$-positive. Observe that $E$ is a $\gC$-encoder if and only if $E^*$ is a $\gC^*$-decoder.

\begin{thm}[Characterization of entanglement annihilation]
\label{thm:EAunderCoding}
Let $V_1$, $V_2$ be vector spaces, $\gC_1 \subset V_1$, $\gC_2 \subset V_2$ be proper cones and $P : V_1 \to V_2$ a linear map. The following are equivalent.
\begin{enumerate}
    \item The map $P$ is entanglement annihilating.
    \item For every $k \in \N$, every $\gC_1$-encoder $E : V_1 \to V_1^{\otimes k}$ and $\gC_2$-decoder $D : V_2^{\otimes k} \to V_2$, the map $D \circ P^{\otimes k} \circ E$ is $(\gC_1,\gC_2)$-positive.
\end{enumerate}
\end{thm}

Before proving Theorem \ref{thm:EAunderCoding}, we introduce a lemma which will be useful later about stability of the class of entanglement annihilating maps.

\begin{lem}[Entanglement annihilation is preserved under separable operations]
For $i \in \{1,2,3,4\}$, consider a proper cone $\gC_i$ in a vector space $V_i$. Consider $(\gC_1,\gC^{\otimes_{\max} k}_2)$-positive maps $E_1,\ldots ,E_N:V_1\ra V^{\otimes k}_2$ and $(\gC^{\otimes_{\min} k}_3,\gC_4)$-positive maps $D_1,\ldots ,D_N:V^{\otimes k}_3\ra V_4$. For any $(\gC_2,\gC_3)$-entanglement annihilating map $P:V_2\ra V_3$, the map $Q:V_1\ra V_4$ given by
\[
Q = \sum^N_{i=1} D_i\circ P^{\otimes k} \circ E_i
\] 
is $(\gC_1,\gC_4)$-entanglement annihilating. 
\label{lem:EntAnnUnderTrans}
\end{lem}

\begin{proof}
For $l\in\N$ and any $i_1,\ldots ,i_l\in \lset 1,\ldots ,N\rset$ we have 
\[
\lb E_{i_1}\otimes \cdots \otimes E_{i_l}\rb\lb \gC^{\otimes_{\max} l}_1\rb \subseteq \lb \gC^{\otimes_{\max} k}_2\rb^{\otimes_{\max} l} = \gC^{\otimes_{\max} kl}_2 ,
\]
and
\[
\lb D_{i_1}\otimes \cdots \otimes D_{i_l}\rb\lb \gC^{\otimes_{\min} kl}_3\rb \subseteq D_{i_1}\lb \gC^{\otimes_{\min} k}_3\rb\otimes_{\min} \cdots \otimes_{\min} D_{i_l}\lb \gC^{\otimes_{\min} k}_3\rb \subseteq \gC^{\otimes_{\min} l}_3.
\]
Now, note that 
\[
Q^{\otimes l} = \sum_{i_1,\ldots ,i_l} \lb D_{i_1}\otimes \cdots \otimes D_{i_l}\rb\circ P^{\otimes kl} \circ \lb E_{i_1}\otimes \cdots \otimes E_{i_l}\rb,
\]
and since $P$ is entanglement annihilating we conclude that 
\[
Q^{\otimes l}\lb \gC^{\otimes_{\max} l}_1\rb \subseteq \gC^{\otimes_{\min} l}_4 .
\]
Since $l\in\N$ was arbitrary, we have shown that $Q$ is entanglement annihilating. 
\end{proof}

\begin{proof}[Proof of Theorem \ref{thm:EAunderCoding}]
Assuming (1), it follows from Lemma \ref{lem:EntAnnUnderTrans} applied with $N=1$ that $D \circ P^{\otimes k} \circ E$ is entanglement annihilating, hence positive.

Conversely, consider a linear map $P:V_1\ra V_2$ that is not entanglement annihilating. Then, there exists a $k\in\N$, an $x\in \gC^{\otimes_{\max} k}_1$, and a $w\in (\gC^*_2)^{\otimes_{\max} k}$ such that $\braket{w}{ P^{\otimes k}(x)}<0$. Define $E:V_1\ra V^{\otimes k}_1$ by $E = x\braket{v}{\cdot}$ for some $v\in \gC^*_1\setminus\lset 0\rset$ and $D:V^{\otimes k}_2\ra V_2$ by $D = y\braket{w}{\cdot}$ for some $y\in \gC_2\setminus\lset 0\rset$. Note that $E$ is a $\gC_1$-encoder and $D$ is $\gC_2$-decoder. However, we have
\[
D\circ P^{\otimes k}\circ E = \braket{w}{P^{\otimes k}(x)} y\braket{v}{\cdot} ,
\]
which is a negative multiple of a non-zero positive map and hence not positive as all involved cones are proper. 
\end{proof}

We conclude this section with two important implications of the previous results for the resilience property of cones. The first one will show that resilience is closed under retracts, and the second that every pair of a resilient cone with any proper cone is resilient as well. 

Say that a cone $\gC'\subset V'$ is a \emph{retract} of a cone $\gC\subset V$ if there exists a $(\gC',\gC)$-positive map $R:V'\ra V$ and a $(\gC,\gC')$-positive map $S:V\ra V'$ such that $\ident_{V'} = S\circ R$. For example, it can be checked easily that $\PSD(\C^{d'})$ is a retract of $\PSD\lb \C^{d}\rb$ if and only if $d'\leq d$, and that $\gL_{n'}$ is a retract of $\gL_n$ if and only if $n' \leq n$. Note also that retracts dualize: If $\gC'$ is a retract of $\gC$, then $\gC'^*$ is a retract of $\gC^*$. We have the following lemma:

\begin{lem}[Resilience is closed under retracts]\label{lem:resilienceRetracts}
Let $\gK\subset W$ be a proper cone and $\gC'\subset V'$ a retract of a proper cone $\gC\subset V$. If the pair $(\gC,\gK)$ \textup{(}or $(\gK,\gC)$\textup{)} is resilient, then the pair $(\gC',\gK)$ \textup{(}or $(\gK,\gC')$\textup{)} is resilient as well. In particular, if $\gC$ is resilient, then $\gC'$ is resilient as well.
\end{lem}

\begin{proof}
Using duality, it is enough to consider the case where $(\gC,\gK)$ is resilient. By definition, we have $\ident_{V'} = S\circ R$ for a $(\gC',\gC)$-positive map $R:V'\ra V$ and a $(\gC,\gC')$-positive map $S:V\ra V'$. Consider an $(\gC',\gK)$-entanglement annihilating map $P:V'\ra W$. By Lemma \ref{lem:EntAnnUnderTrans}, the map $P\circ S:V\ra W$ is $(\gC,\gK)$-entanglement annihilating and hence entanglement breaking by resilience of $(\gC,\gK)$. We conclude that $P = P\circ S\circ R$ is entanglement breaking as well, and thus the pair $(\gC',\gK)$ is resilient.
\end{proof}

As a consequence of Lemma \ref{lem:resilienceRetracts} we conclude that resilience of $\PSD(\C^{d})$ implies resilience of $\PSD(\C^{d'})$ when $d'\leq d$, and that resilience of $\gL_n$ implies resilience of $\gL_{n'}$ whenever $n'\leq n$. Moreover, it can be checked \cite[Proposition S7]{aubrun2019universal} that the Lorentz cone $\gL_n$ is a retract of $\PSD(\C^d)$ for $d=2^n$, and therefore resilience of $\PSD(\C^{2^n})$ would imply resilience of $\gL_n$. We will show the latter in a different way. 

We will finish this section with another basic property of resilient cones:

\begin{lem}[Resilience implies resilience of pairs]\label{lem:resilienceImpliesResOfPair}
Let $\gC \subset V$ be a proper cone. Then the following are equivalent:
\begin{enumerate}
\item The cone $\gC$ is resilient. 
\item The pair $(\gC,\gC')$ is resilient for every proper cone $\gC'\subset V'$.
\item The pair $(\gC',\gC)$ is resilient for every proper cone $\gC'\subset V'$
\end{enumerate}
\end{lem}
\begin{proof}
It is clear that the second statement implies the first. To show the other direction assume that $\gC$ is resilient and that for some proper cone $\gC'\subset V'$ the pair $(\gC,\gC')$ is not resilient. Then, there exists a $(\gC,\gC')$-entanglement annihilating map $P:V\ra V'$ that is not $(\gC,\gC')$-entanglement breaking. By Lemma \ref{lem:dualityEB} there is a $(\gC',\gC)$-positive map $Q:V'\ra V$ such that $\Tr[Q \circ P] < 0$. Again by Lemma \ref{lem:dualityEB}, we conclude that $Q \circ P:V\ra V$ is not entanglement breaking, but by Lemma \ref{lem:EntAnnUnderTrans} (for $k=1$) it is entanglement annihilating. This contradicts the assumption. Equivalence of the first and third statements follows in a similar way. 
\end{proof}

\section{Symmetrization of positive maps between cones}
\label{sec:Symmetrization}

To show that a pair of cones $(\gC_1,\gC_2)$ is resilient, it is, a priori, necessary to check whether every entanglement annihilating map is entanglement breaking. In this section, we present two types of cones for which the resilience question can be reduced to entanglement annihilating maps with particular properties.  

\subsection{Cones with a symmetric base}
\label{sec:ConeWithSymBaseRed}

Consider a cone $\gC_X\subset \R\oplus X$ associated with a finite-dimensional normed space $X$ (see Section~\ref{sec:properConesNormedSpaces}). We show that the resilience of $\gC_X$ can be checked using only central maps. Recall that a central map has the form $\alpha \oplus P$ for $P:X \to X$, and that $\alpha \oplus P$ is $\gC_X$-entanglement breaking if and only if $\|P\|_{N(X \to X)} \leq \alpha$.

\begin{thm}
\label{thm:resilienceVstildePalpha}
For a finite-dimensional normed space $X$ the following are equivalent:
\begin{enumerate}
\item The cone $\gC_X$ is resilient.
\item Every $\gC_X$-entanglement annihilating central map is entanglement breaking.
\end{enumerate} 
\end{thm}

\begin{proof}
It is obvious that (1) implies (2). Conversely, assume that $\gC_X$ is not resilient and let $R : \R \oplus X \to \R \oplus X$ a map which is $\gC_X$-entanglement annihilating and not $\gC_X$-entanglement breaking. By Lemma \ref{lem:dualityEB}, there is a $\gC_X$-positive map $Q$ such that $\Tr\lbr Q\circ R\rbr<0$. Let $A : \R \oplus X \to \R \oplus X$ be the automorphism of $\gC_X$ defined by $A(t,x)=(t,-x)$ and 
\[S= \frac{1}{2} \left( Q \circ R + A \circ Q \circ R \circ A \right). \]
It is easy to check that $S$ is a central map. Since $A$ and $Q$ are $\gC_X$-positive, we conclude by Lemma \ref{lem:EntAnnUnderTrans} that $S$ is entanglement annihilating. Moreover, since
\[
\Tr\lbr S \rbr = \Tr\lbr Q\circ R\rbr<0 ,
\]
an application of Lemma \ref{lem:dualityEB} shows that $\tilde{P}_\alpha$ is not entanglement breaking.
\end{proof}

We will apply this theorem in Section \ref{sec:ConeSymBase} to relate resilience of the cone $\gC_X$ to properties of the normed space $X$. In the next section, we will consider cones with enough symmetries to reduce resilience to a much smaller class of maps. 

\subsection{Cones with enough symmetries and twirling to isotropic maps}

\label{sec:EnoughSymm}

Let $V$ be an $n$-dimensional Euclidean space, which we identify with $\R^n$. Given a convex body $B \subset V$, we say that an orthogonal map $g \in O_n$ is an isometry of $B$ if $g(B)=B$. The set of isometries of $B$, which we denote $\Iso(B)$, is a closed subgroup of $O_n$. 
We say that $B$ has \emph{enough symmetries} if $\Iso(B)' = \R 1$; here $G'$ denotes the commutant of $G$, i.e., the set of linear maps $S : V \to V$ such that $gS=Sg$ for every $g \in G$. Note that a convex body $B$ with enough symmetries has centroid at the origin; in particular $0 \in \inter(B)$. Slightly abusing notation, we will sometimes say that a cone $\gC$ has enough symmetries if there is a base $B$ with enough symmetries satisfying $\gC=\gC_B$. The family of cones with enough symmetries includes the cones $\gC_{\ell_p}$ and more generally the cones $\gC_X$ for normed spaces with enough symmetries~\cite{tomczak1989banach}. Moreover, it also contains the positive semidefinite cones $\PSD(\C^d)$, a fact that has often been exploited in entanglement distillation~\cite{werner1989quantum,divincenzo2000evidence} and which inspired the techniques developed here.

We denote by $\gC_B \subset \R^{n+1}$ the cone over $B$. If $g \in \Iso(B)$ is an isometry, we denote by $\tilde{g} : \R^{n+1} \to \R^{n+1}$ the automorphism of $\gC_B$ defined by $\tilde{g}(t,x)=(t,gx)$ for $(t,x) \in \R \oplus \R^n$. Denote by $\pi_1$ and $\pi_2$ the orthogonal projections defined as
\[ \pi_1(t,x)=(t,0), \ \ \ \pi_2(t,x)=(0,x) \]
for $(t,x) \in \R \oplus \R^n$.

Let $B \subset \R^n$ be a convex body with enough symmetries. We define the \emph{twirling operator} $\tau$ as follows: If $L: \R^{n+1} \ra \R^{n+1}$ is a linear map, then $\tau[L] : \R^{n+1} \ra \R^{n+1}$ if defined as
\[ \tau[L] := \int_{\Iso(B)} \tilde{g}^{-1} \circ L \circ \tilde{g} \, \mathrm{d}g \]
where the integral is with respect to the normalized Haar measure on $\Iso(B)$.

\begin{prop}[Twirling and isotropic maps]
Let $B \subset \R^n$ be a convex body with enough symmetries and $L: \R^{n+1} \ra \R^{n+1}$ a linear map. Then
\[ \tau[L] = \alpha \pi_1 + \beta \pi_2 ,\]
where $\alpha = \Tr(\pi_1 L \pi_1)$ and $\beta = \frac{1}{n} \Tr (\pi_2 L \pi_2).$
\label{prop:TwirlAndIso}
\end{prop}

\begin{proof}
Consider the block matrix $\tau[L]=\begin{pmatrix} \alpha & y^T \\ x & A \end{pmatrix}$ with $\alpha \in \R$, $x,y \in \R^n$ and $A \in \M_n(\R)$. For every $g \in \Iso(B)$, $\tau[L]$ commutes with $\tilde{g}$ by invariance of the Haar measure. It follows that $gx=x$, $y^T g =g$ and $Ag=gA$. Since $B$ has enough symmetries, we have $x=y=0$ and $A=\beta \id$ for some $\beta \in \R$. We proved that $\tau[L] = \alpha \pi_1 + \beta \pi_2$ and the values for $\alpha$ and $\beta$ are easily computed.
\end{proof}

For $\alpha,\beta\in \R$ and $n\in\N$ we define the \emph{isotropic map} $I_{\alpha,\beta} : \R^{n+1} \to \R^{n+1}$ by
\begin{equation}
I_{\alpha,\beta} = \alpha \pi_1 + \beta \pi_2 .
\label{equ:IsoMap}
\end{equation}
We will sometimes use the notation $I_\alpha := I_{\alpha,1}$ to denote a normalized isotropic maps. By Proposition \ref{prop:TwirlAndIso}, applying the twirling operator always produces an isotropic map. The following lemma characterizes some elementary properties of the isotropic maps: 

\begin{lem}[Properties of isotropic maps]
Let $B \subset \R^n$ be a convex body with enough symmetries. For $\alpha, \beta \in \R$, consider the isotropic map $I_{\alpha,\beta} = \alpha \pi_1 + \beta \pi_2$. We have the following: 
\begin{enumerate}
\item The isotropic map $I_{\alpha,\beta}$ is $\gC_B$-positive if and only if $\alpha \geq 0$ and $\beta B \subseteq \alpha B$. 
\item The map $I_{\alpha,\beta}$ is $\gC_B$-entanglement breaking if and only if $\alpha \geq 0$ and $\beta\gamma \geq - \alpha/n$ for every $\gamma\in\R$ such that $\gamma B \subseteq B$. In particular, $\beta \geq - \alpha/n$ whenever $I_{\alpha,\beta}$ is $\gC_B$-entanglement breaking.
\end{enumerate}
If in addition $-B=B$, then we have the following refinement:
\begin{enumerate}
\item[(3)] The map $I_{\alpha,\beta}$ is $\gC_B$-positive if and only if $|\beta|\leq \alpha$. 
\item[(4)] The map $I_{\alpha,\beta}$ is $\gC_B$-entanglement breaking if and only if $|\beta| \leq \alpha/n$. 
\end{enumerate}
\label{lem:PropsIso}
\end{lem}

\begin{proof}
Given $(t,x) \in \R^{n+1}$, we have $I_{\alpha,\beta}(t,x)=(\alpha t,\beta x)$. The condition $\alpha \geq 0$ is necessary for $I_{\alpha,\beta}$ to be positive. Since $(t,x) \in \gC_B$ if and only if $t \geq 0$ and $x \in t B$, (1) follows. If $B$ is symmetric, this is clearly equivalent to $|\beta| \leq \alpha$.

By duality (Lemma \ref{lem:dualityEB}), the map $I_{\alpha,\beta}$ is $\gC_B$-entanglement breaking if and only if $\tr[Q \circ I_{\alpha,\beta}] \geq 0$ for every $Q \in \mathcal{P}(\gC_B)$. By Proposition \ref{prop:TwirlAndIso}, we have $\tau[Q] = I_{\gamma,\delta}$ for some $\gamma,\delta \in \R$. Using cyclicity of the trace, we compute
\[ \Tr[Q \circ I_{\alpha,\beta}] = \Tr[ Q \circ \tau[I_{\alpha,\beta}] ] = \Tr[ \tau[Q] \circ I_{\alpha,\beta} ] = \Tr [I_{\gamma,\delta} \circ I_{\alpha,\beta}] = \alpha \gamma + n \beta \delta.\]
Since $Q \in \mathcal{P}(\gC_B)$ implies $\delta B \subset \gamma B$, (2) follows. In the symmetric case, it is enough to consider the extremal values $\delta = \pm \gamma$, giving (4).
\end{proof}

By applying the twirling technique, we will now reduce the question of resilience for cones with enough symmetries to determining whether every entanglement annihilating isotropic map is entanglement breaking. For this we need to ensure that the isotropic map obtained from twirling a positive and non-entanglement breaking map is non-entanglement breaking itself. We start with an easy lemma: 

\begin{lem}
Let $B \subset \R^n$ be a convex body with enough symmetries and $P \in \mathcal{P}(\gC_B)$ a positive map such that $\Tr[P] < 0$. Then the isotropic map $\tau[P]$ is not $\gC_B$-entanglement breaking.
\label{lem:BasicLemNTr}
\end{lem}

\begin{proof}
By Proposition \ref{prop:TwirlAndIso}, we have
$\tau[P] = \alpha \pi_1 + \beta \pi_2 $
with $\alpha = \Tr (\pi_1 L \pi_1)$ and $\beta = \frac{1}{n} \Tr (\pi_2 L \pi_2)$. We also have $\Tr[P] = \Tr(\pi_1 L) + \Tr (\pi_2 L) = \alpha +n \beta$. It follows that $\alpha + n \beta <0$ and therefore $\beta < -\alpha/n$. By Lemma \ref{lem:PropsIso}, the map $\tau[P]$ is not entanglement breaking.
\end{proof}

We can now prove the following theorem: 

\begin{thm}[Twirling with filter]
Let $B \subset \R^n$ be a convex body with enough symmetries and $\gC \subset V$ any proper cone.
\begin{enumerate}
\item If $P \in \mathcal{P}(\gC_B,\gC)$ is a positive map which is not $(\gC_B,\gC)$-entanglement breaking, there exists a positive map $Q \in \mathcal{P}(\gC,\gC_B)$ such that the isotropic map $\tau \lbr Q\circ P\rbr$ is $\gC_B$-positive and not $\gC_B$-entanglement breaking. 
\item If $P \in \mathcal{P}(\gC,\gC_B)$ is a positive map which is not $(\gC,\gC_B)$-entanglement breaking, there exists a positive map $R \in \mathcal{P}(\gC_B,\gC)$ such that the isotropic map $\tau \lbr P\circ R\rbr$ is $\gC_B$-positive and not $\gC_B$-entanglement breaking.
\end{enumerate}
\label{thm:TwirlingReduction}
\end{thm}

\begin{proof}
By duality (Lemma \ref{lem:dualityEB}), there exists a positive map $Q \in \mathcal{P}(\gC,\gC_B)$ such that $\Tr\lbr Q\circ P \rbr <0$.
Applying Lemma \ref{lem:BasicLemNTr} to the $\gC_B$-positive map $Q\circ P$ shows that the isotropic map $\tau \lbr Q\circ P\rbr$ is not $\gC_B$-entanglement breaking. Clearly, $\tau \lbr Q\circ P\rbr$ is $\gC_B$-positive as a twirl of a positive map. This finishes the proof of the first part, and the second part is proved similarly. \end{proof}

We have the following corollary: 

\begin{cor}[Resilience of cones with enough symmetries]
Let $B \subset \R^n$ be a convex body with enough symmetries. The following are equivalent:
\begin{enumerate}
\item The cone $\gC_B$ is resilient.
\item Every isotropic map on $\R^{n+1}$ which is $\gC_B$-entanglement annihilating is also $\gC_B$-entanglement breaking.
\end{enumerate}
\label{cor:RedToIso}
\end{cor}

\begin{proof}
It is obvious that the first statement implies the second. For the reverse direction assume that the second statement holds, but that there exists a map $P:\R^{n+1}\ra \R^{n+1}$ which is entanglement annihilating map but not entanglement breaking. By Theorem \ref{thm:TwirlingReduction}, there exists a positive map $Q \in \mathcal{P}(\gC_B,\gC_B)$ such that $\tau \lbr Q\circ P\rbr$ is not entanglement breaking. Applying Lemma 
\ref{lem:EntAnnUnderTrans} for $k=1$ shows that the isotropic map $\tau \lbr Q\circ P\rbr$ is entanglement annihilating, a contradiction. 
\end{proof}

\section{Resilience of Lorentz cones and proof of Theorem \ref{thm:Main1}} \label{sec:lorentz}

We will prove the following:

\begin{thm} \label{thm:Lorentz}
For every $n \geq 1$, the Lorentz cone $\gL_n$ is resilient.
\end{thm}

Using Theorem \ref{thm:Lorentz}, the proof of Theorem \ref{thm:Main1} is easy.

\begin{proof}[Proof of Theorem \ref{thm:Main1}]
Combine Theorem \ref{thm:Lorentz} and Lemma \ref{lem:resilienceImpliesResOfPair}. 
\end{proof}

To prove Theorem \ref{thm:Lorentz} we have two strategies: In Section \ref{sec:Lorentz-composition} we generalize the approach used to prove the resilience of $\gL_3$ (by the distillation protocol of entangled qubit states~\cite{bennett1996purification}) to prove the resilience of $\gL_n$ for $n \leq 9$. This restriction is explained by the fact that our construction relies on real composition algebras, which only exist in small dimensions. However, in Section \ref{sec:Lorentz-RH}, we present a different family of protocols, which allow to prove resilience of $\gL_n$ for every $n\in\N$.

\subsection{Resilience of certain Lorentz cones from composition algebras}

\label{sec:Lorentz-composition}

A \emph{real composition algebra} is a triple $\A=(V,*,q)$, where $(V,*)$ is a finite-dimensional unital algebra over the reals and $q$ a nondegenerate (i.e., full rank) quadratic form satisfying the condition
\begin{equation}
\label{eq:composition}    
 q(x * y) = q(x) q(y)
\end{equation}
for every $x,y \in V$. 

A complete classification of real composition algebras is available since works by Hurwitz and Cayley. Let $\A = (V,*,q)$ be a real composition algebra. The classification theorem (see, e.g.,~\cite[Theorem 1.10]{roos08}) asserts that $\A$ is isomorphic to one of the 7 real composition algebras listed below
\begin{itemize}
    \item If $q$ is positive definite, then $(V,*)$ is a \emph{division algebra}. It is isomorphic to either the real numbers $\R$, the complex numbers $\C$, the quaternions $\qua$ or the octonions $\oc$. For $\A \in \{\R,\C,\qua,\oc\}$, we abusively denote by $\A$ the composition algebra $(\A,*,q)$, where $*$ is the usual multiplication and $q$ the square of the usual norm.
    \item Otherwise, $(V,*)$ is a \emph{split algebra}. It is isomorphic to either the split complex numbers $\C'$, the split quaternions $\qua'$ or the split octonions $\oc'$. In the following, we only consider the split complex numbers $\C'$, which are defined as the real composition algebra $(\R^2,*,q)$ where
    \begin{equation} 
    \label{eq:split-complex}
    (x,y)*(x',y') = (xx'+yy',xy'+x'y) \end{equation}
    and $q(x,y)=x^2-y^2$.
\end{itemize}

Given a composition algebra $\A = (V,*,q)$ we denote by $m_\A:V\otimes V \ra V$ the multiplication tensor given by $m_\A(x\otimes y) = x* y$, and extended linearly.
When $\A$ is a division algebra (i.e., $\A \in \{\R,\C,\qua,\oc\}$), we may identify $V$ and $V^*$ using the inner product derived from $q$ and therefore consider the adjoint $m_\A^*$ as an operator from $V$ to $V \otimes V$. In each case, one checks the relation
\begin{equation}
\label{eq:multiplication-tensor}    
 m_\A \circ m_\A^* = \dim(V) \one_V . 
\end{equation}

Consider two composition algebras 
$\A_1=(V_1,*,q_1)$, $\A_2=(V_2,*,q_2)$ and set $V=V_1 \oplus V_2$. One defines the direct sum $m_{\A_1} \oplus m_{\A_2} : V \otimes V \to V$ as
\[
(m_{\A_1}\oplus m_{\A_2})((x_1,x_2)\otimes (y_1,y_2)) = (m_{\A_1}(x_1,y_1),m_{\A_2}(x_2,y_2))
\]
for $x_1,y_1$ in $\A_1$ and $x_2,y_2$ in $\A_2$.

We start with a lemma: 

\begin{lem} \label{lem:L-positive-map}
Consider $\A_1 \in \lset \R,\C'\rset$ and $\A_2\in \lset \R,\C,\qua,\oc\rset$, with respective quadratic forms $q_1$ and $q_2$. The cone 
\begin{equation} 
\label{eq:def-L}
\gL = \{ (x_1,x_2) \in \A_1 \oplus \A_2 \st q_2(x_2) \leq q_1(x_1) \}
\end{equation}
 is isomorphic to $\gL_N$ for $N=\dim(\A_1)+\dim(\A_2)-1$ and satisfies
\[
(m_{\A_1}\oplus m_{\A_2})\lb \gL \otimes_{\min} \gL \rb\subseteq \gL .
\]
\end{lem}

\begin{proof}
Since $q_2$ is positive definite and $q_1$ has signature either $(1,0)$ or $(1,1)$, it is immediate to check that $\gL$ is isomorphic to $\gL_N$. Consider now $(x_1,x_2)$ and $(y_1,y_2)$ in $\gL$.
We compute, using the property \eqref{eq:composition}
\[ q_1(m_{\A_1}(x_1,y_1)) 
= q_1(x_1)q_1(y_1) \geq q_2(x_2)q_2(y_2) = q_2(m_{\A_2}(x_2,y_2)).
\]
We conclude that $(m_{\A_1}\oplus m_{\A_2})\lb (x_1,x_2)\otimes (y_1,y_2)\rb\in \gL$ and the result follows.
\end{proof}

\begin{prop} \label{prop:resilience-L}
Consider $\A_1 \in \lset \R,\C'\rset$ and $\A_2\in \lset \R,\C,\qua,\oc\rset$, with respective quadratic forms $q_1$ and $q_2$, and the cone $\gL$ defined in \eqref{eq:def-L}. 
Then the cone $\gL$ is resilient.
\end{prop}

We conclude from Lemma \ref{lem:L-positive-map} and Proposition \ref{prop:resilience-L} that the Lorentz cone $\gL_N$ is resilient for $N \in \{1,2,3,4,5,8,9\}$. Moreover, the cases $N \in \{6,7\}$ are covered by the fact that the resilience of $\gL_N$ implies the resilience of $\gL_n$ for every $n \leq N$.

\begin{proof}
By Corollary \ref{cor:RedToIso}, it is enough to show that
an isotropic map (i.e., a map of the form $I_{\alpha,\beta} = \alpha \pi_1 + \beta \pi_2$) which is entanglement annihilating is also entanglement breaking. Using homogeneity  and the fact that the Lorentz cone has a symmetric base, it suffices to prove this for $\alpha = 1$ and $\beta \geq 0$. 

Let $I_{\alpha,\beta} = \alpha \pi_1 + \beta \pi_2$ be an isotropic map which is entanglement annihilating. By Lemma \ref{lem:L-positive-map}, the map $m_{\A_1} \oplus m_{\A_2}$ is $(\gL \tmin \gL,\gL)$-positive. By Lemma \ref{lem:dual},  its adjoint $(m_{\A_1} \oplus m_{\A_2})^*$ is therefore $(\gL,\gL \tmax \gL)$-positive. We conclude by Lemma \ref{lem:EntAnnUnderTrans} that the map
\[  
J := (m_{\A_1} \oplus m_{\A_2}) \circ I_{\alpha,\beta}^{\otimes 2} \circ 
(m_{\A_1} \oplus m_{\A_2})^*
\]
is entanglement annihilating. Finally, the twirled map $\tau [J]$ is also entanglement annihilating (the fact that the class of entanglement annihilating maps is stable under twirling is a consequence of Lemma \ref{lem:EntAnnUnderTrans}; note that in the definition of the twirling operator we may, using Caratheodory's theorem, replace the Haar measure by a suitable finite measure). The map $\tau[J]$ is isotropic and equals $I_{\alpha',\beta'}$ with $\alpha'$, $\beta'$ given by the following formula
\begin{equation}
\label{eq:alpha'}    
(\alpha',\beta') = \begin{cases} (\alpha^2, n \beta^2) & \textnormal{ if } \A_1 = \R \\ (\alpha^2 + \beta ^2, \frac{2\alpha\beta + \beta^2 n^2}{n+1} )& \textnormal{ if } \A_1 = \C' \end{cases}
\end{equation} 
where $n = \dim(\A_2)$. We only justify this formula when $\A_1=\C'$, the case $\A_1=\R$ being similar and simpler. Let $(e_0,e_1)$ and $(f_1,\dots,f_n)$ the canonical bases of $\A_1$ and $\A_2$ respectively. We have
\[ I_{\alpha,\beta} = \alpha \ketbra{e_0} + \beta \ketbra{e_1} + \beta \sum_{i=1}^n \ketbra{f_i}.\]
Using \eqref{eq:split-complex} and \eqref{eq:multiplication-tensor}, it follows that
\[ J = (\alpha^2 + \beta^2) \ketbra{e_0} + 2 \alpha \beta \ketbra{e_1} + \beta^2 n \sum_{k=1}^n \ketbra{f_k} \]
applying the twirling operator
yields \eqref{eq:alpha'}. In particular we have
\[ I_{1,\beta} \textnormal{ entanglement annihilating} \Longrightarrow I_{1,f(\beta)} \textnormal{ entanglement annihilating} \]
where
\[
f(\beta) = \begin{cases} n\beta^2 , &\text{ if } \A_1 = \R \\ \frac{\beta (2 + n^2\beta)}{(n+1)(1 + \beta^2)} , &\text{ if } \A_1 = \C'\end{cases}.
\]
Let $\beta_0$ the largest $\beta > 0$ such that $I_{1,\beta}$ is entanglement annihilating. We have
$f(\beta_0) \leq \beta_0$, which implies that 
\[\beta_0 \leq \begin{cases}
\frac{1}{n} & \text{ if } \A_1 = \R \\
\frac{1}{n-1} & \text{ if } \A_1 = \C'.
\end{cases} \]
In both cases, it follows that whenever $I_{1,\beta}$ is entanglement annihilating, then $\beta \leq (\dim \gL -1)^{-1}$ and therefore $I_{1,\beta}$ is entanglement breaking by Lemma \ref{lem:PropsIso}. We conclude that $\gL$ is resilient.
\end{proof}

\subsection{Resilience of all Lorentz cones}

\label{sec:Lorentz-RH}

Given an integer $n$, let $N(n)$ be the minimal $N$ such that there exists an $n$-dimensional subspace $E \subset \cM_N(\R)$ in which every matrix is a multiple of an orthonormal matrix. For our purposes, we only need to know that $N(n)$ is finite for every integer $n$. The value of $N(n)$ is known and related to the Radon--Hurwitz number (see, e.g., \cite[Theorem 11.4]{aubrun2017alice}).

We need the following lemma. Although it is contained as an exercise in \cite{aubrun2017alice}, we include here the proof for completeness.

\begin{lem}
Fix integers $n, k \geq 1$. There is an element $z_{n,k}$ in the Euclidean space $(\R^n)^{\otimes k}$ with the following properties.
\begin{enumerate}
\item  For every $x_1,\ldots ,x_k\in \R^n$, we have
\[
\braket{x_1\otimes \dots \otimes x_k}{z_{n,k}}\leq \|x_1\|_2\cdots \| x_k\|_2 ;
\]
\item we have 
\[ \|z_{n,k}\|^2_2 \geq \frac{n^k}{N(n)}. \]
\end{enumerate}
\label{lem:MagicalZ}
\end{lem}

\begin{proof}
Set $N=N(n)$. Let $E \subset \cM_N(\R)$ be an $n$-dimensional subspace in which every matrix is a multiple of an orthogonal matrix. Let $(A_1,\dots,A_n)$ be a basis of $E$ which is orthonormal with respect to the inner product $(A,B) \mapsto \frac{1}{N} \Tr (AB)$.  
Let $\Theta:\R^n\ra \M_N(\R)$ denote the function 
\[
\Theta(x) = \sum^n_{i=1} x_i A_i,
\]
and note that $\Theta(x)^T\Theta(x)=\|x\|^2_2\one_N$ for any $x\in\R^n$. For $i,j\in\lset 1,\ldots ,N\rset$ we define
\begin{equation}
z(i,j) = \sum_{l_1,\dots,l_k=1}^n \lbr\Theta(e_{l_1})\cdots \Theta(e_{l_k})\rbr_{ij} e_{l_1}\otimes \cdots \otimes e_{l_k} \in (\R^n)^{\otimes k}.
\label{equ:MagicZ}
\end{equation}
We first check that each such tensor satisfies the first conclusion of Lemma \ref{lem:MagicalZ}.
For $x_1,\ldots ,x_k\in \R^n$ an easy computation shows that
\[
\braket{x_1\otimes \ldots \otimes x_k}{z_k(i,j)} = \lbr\Theta(x_1)\cdots \Theta(x_k)\rbr_{ij} .
\]
Using first that $e_i\braket{e_i}{\cdot}\leq \one_N$ in the Loewner order for any $i\in\lset 1,\ldots, n\rset$ and that $\M_N\ni X\mapsto A^TXA$ preserves the Loewner order for any matrix $A\in\M_N$, and then the properties of $\Theta(\cdot)$ outlined above, we find that 
\begin{align*}
\lbr\Theta(x_1)\cdots \Theta(x_k)\rbr^2_{ij} &\leq \braket{e_j}{\Theta(x_k)^T\cdots \Theta(x_1)^Te_i}\braket{e_i}{\Theta(x_1)\cdots \Theta(x_k)e_j} \\
&\leq \braket{e_j}{\Theta(x_1)\cdots \Theta(x_k)\Theta(x_k)^T\cdots \Theta(x_1)^Te_j} \\
&\leq \|x_1\|^2_2\cdots \| x_k\|^2_2.
\end{align*}
We now compute  
\begin{align*}
\sum_{i,j=1}^N & \|z(i,j)\|_2^2 \\
&= \sum_{i,j=1}^N \sum_{l_1,\ldots ,l_k=1}^n \braket{e_j}{\Theta(e_{l_k})^T\cdots \Theta(e_{l_1})^Te_i}\braket{e_i}{\Theta(e_{l_1})\cdots \Theta(e_{l_k})e_j} \\
&= \Tr\lbr\Theta(e_{l_k})^T\cdots \Theta(e_{l_1})^T \Theta(e_{l_1})\cdots \Theta(e_{l_k})\rbr= Nn^k,
\end{align*}
where we used the properties of $\Theta(\cdot)$. It follows that there exists $i_0,j_0 \in \{1,\dots,N\}$ such that $\|z(i_0,j_0)\|_2^2 \geq n^k/N$. Therefore, the choice $z_{n,k}=z(i_0,j_0)$ satisfies both properties from Lemma \ref{lem:MagicalZ}.
\end{proof}

We are now in the position to show that all Lorentz cones are resilient.

\begin{proof}[Proof of Theorem \ref{thm:Lorentz}]
The Lorentz cone $\gL_n\subset \R^{n+1}$ can be identified with the cone over the unit ball in the Euclidean space $\R^n$, which has enough symmetries. We use the results from section \ref{sec:EnoughSymm} and consider the isotropic maps
\[ I_{\alpha,\beta}:= \alpha \pi_1 + \beta \pi_2 \]
for $\alpha >0$ and $\beta \in \R$ such that $I_{\alpha,\beta}$ is entanglement annihilating. Fix $k\in\N$ and let $z_{n,k} \in (\R^n)^{\otimes k}$ be the tensor given by Lemma \ref{lem:MagicalZ}. We consider $z_{n,k}$ as an element of $(\R^{n+1})^{\otimes k}$ by identifying $\R^n$ with the subspace $\{ (0,x) \st x \in \R^n \} \subset \R^{n+1}$. We claim that the tensors $z^+$ and $z^-$ defined by
\[ z^\pm = e_0^{\otimes k} \pm z_{n,k} \]
belong to $\gL_n^{\tmax k}$. To check this claim, consider elements $x_1=(t_1,y_1), \dots, x_k=(t_k,y_k) \in \gL_n$ (identified with $\gL_n^*$) and find that 
\begin{align*} \langle x_1 \otimes \dots \otimes x_k, z^\pm \rangle & =  t_1 \cdots t_k \pm \langle y_1 \otimes \dots \otimes y_k, z_{n,k} \rangle \\
& \geq t_1\cdots t_k - \|y_1\|_2 \dots \|y_k\|_2 \\
& \geq 0. \end{align*}
Since $I_{\alpha,\beta}$ is entanglement annihilating, we have $I_{\alpha,\beta}(z^+) = \alpha^k e_0^k + \beta^k z_{n,k} \in \gL_n^{\tmin k}$ and therefore
\[ 0 \leq \langle z^\pm, I_{\alpha,\beta}(z^+) \rangle =
\alpha^k \pm \beta^k \|z_{n,k}\|_2^2 .\]
If we choose the sign $\pm$ such that $\pm \beta^k \leq 0$, we have 
\[ \alpha^k \geq \beta^k \|z_{n,k}\|_2^2 \geq \frac{|\beta|^kn^k}{N(n)} \] 
and the inequality $\alpha \geq |\beta|n$ follows by taking $k$ to infinity. By Proposition \ref{lem:PropsIso}, the map $I_{\alpha,\beta}$ is entanglement breaking. The argument above shows that every entanglement annihilating isotropic map is entanglement breaking. By Corollary \ref{cor:RedToIso}, this implies that $\gL_n$ is resilient.
\end{proof}

\section{Factorization and breaking entanglement with some cone}\label{sec:fact}

To gain a better understanding of the structure of entanglement annihilating maps with respect to cones $\gC_1$ and $\gC_2$, we can study their properties relative to a third cone $\gK$. Although it might be difficult to show that all entanglement annihilating maps are entanglement breaking and thereby proving resilience of the pair $(\gC_1, \gC_2)$, it turns out that all entanglement annihilating maps break entanglement with resilient cones. After explaining the general theory, we will study the special case of the positive semidefinite matrices $\PSD_d$ and the Lorentz cones $\gL_n$, which we know to be resilient from Section \ref{sec:lorentz}. In this setting, we will establish a generalization of the reduction criterion from entanglement distillation.

\subsection{General theory}
\label{sec:GenFactTheory}

We first introduce two cones of maps associated to a proper cone $\gK$. If $\gK = \R^+$ is a $1$-dimensional cone, maps which factor through $\gK$ are exactly entanglement breaking maps.

\begin{defn}
Let $\gC_1\subset V_1$, $\gC_2\subset V_2$, $\gK\subset V_3$ denote proper cones. 
\begin{enumerate}
	\item We say that a $(\gC_1,\gC_2)$-positive map $P:V_1\ra V_2$ \emph{factors through $\gK$} if it can be written as a finite sum $\sum_{i} S_i\circ R_i$ with $(\gK,\gC_2)$-positive maps $S_i$ and $(\gC_1,\gK)$-positive maps $R_i$. We denote the cone of $(\gC_1,\gC_2)$-positive maps factoring through $\gK$ by $\mathcal{F}_{\gK}(\gC_1,\gC_2)$.
	\item We say that a $(\gC_1,\gC_2)$-positive map $P:V_1\ra V_2$ \emph{breaks the entanglement with $\gK$} if $S\circ P\circ R$ is $\gK$-entanglement breaking for any $(\gC_2,\gK)$-positive map $S:V_2\ra V_3$ and any $(\gK,\gC_1)$-positive map $P:V_3\ra V_1$. We denote the cone of $(\gC_1,\gC_2)$-positive maps breaking the entanglement with $\gK$ by $\mathcal{EB}_\gK(\gC_1,\gC_2)$.
\end{enumerate}
\end{defn}

The following lemma follows immediately from the canonical isomorphism between linear maps and tensors, and by using the duality in Lemma \ref{lem:dualityEB}.

\begin{lem}\label{lem:KEBEasyProperties}
Consider proper cones $\gC_1\subset V_1$, $\gC_2\subset V_2$, $\gK\subset V_3$ and a $(\gC_1,\gC_2)$-positive map $P:V_1\ra V_2$. The following are equivalent:
\begin{enumerate}
\item We have $P\in \mathcal{EB}_\gK(\gC_1,\gC_2)$.
\item We have $\lb\ident_{V_3}\otimes P\rb\lb \gK^*\otimes_{\max} \gC_1\rb\subseteq \gK^*\otimes_{\min} \gC_2$.
\item The composition $P\circ R$ is $(\gK,\gC_2)$-entanglement breaking for every $(\gK,\gC_1)$-positive map $R:V_3\ra V_1$.
\item The composition $S\circ P$ is $(\gC_1,\gK)$-entanglement breaking for every $(\gC_2,\gK)$-positive map $S:V_2\ra V_3$.
\end{enumerate} 
\end{lem}

The following theorem shows that the two cones introduced above are dual with respect to the Hilbert--Schmidt inner product. 

\begin{thm}[Maps breaking $\gK^*$-entanglement]
\label{thm:DualOfFactor}
For proper cones $\gC_1\subset V_1$, $\gC_2\subset V_2$, $\gK\subset V_3$  we have 
\[
\mathcal{F}_{\gK}(\gC_1,\gC_2)^\circ = \mathcal{EB}_{\gK^*}(\gC^*_1,\gC^*_2).
\]
\end{thm}

\begin{proof}
To show that $\mathcal{EB}_{\gK^*}(\gC^*_1,\gC^*_2)\subseteq \mathcal{F}_{\gK}(\gC_1,\gC_2)^\circ$ consider $P\in \mathcal{EB}_{\gK^*}(\gC^*_1,\gC^*_2)$. Using cyclicity of the trace and Lemma 
\ref{lem:dualityEB} we find 
\[
\braket{S\circ R}{P} = \Tr\lbr R^*\circ S^* \circ P\rbr = \Tr\lbr S^* \circ P\circ R^*\rbr \geq 0,
\]
for any $S\in \mathcal{P}(\gK,\gC_2)$ and $R\in \mathcal{P}(\gC_1,\gK)$. This shows that $P\in \mathcal{F}_{\gK}(\gC_1,\gC_2)^\circ$. 

To show that $\mathcal{EB}_{\gK^*}(\gC^*_1,\gC^*_2)\supseteq \mathcal{F}_{\gK}(\gC_1,\gC_2)^\circ$ assume that $P:V_1\ra V_2$ satisfies $P\notin \mathcal{EB}_{\gK^*}(\gC^*_1,\gC^*_2)$. By definition there exist $S_1\in \mathcal{P}(\gK,\gC_2)$ and $R\in \mathcal{P}(\gC_1,\gK)$ such that $S_1^* \circ P\circ R^*$ is not $\gK^*$-entanglement breaking. By duality of entanglement breaking maps and positive maps there exists an $S_2\in \mathcal{P}(\gK,\gK)$ such that 
\[
\braket{S_2}{S_1^* \circ P\circ R^*} = \Tr\lbr S^*_2\circ S_1^*\circ P\circ R^*\rbr = \braket{S_1\circ S_2\circ R}{P} <0 .
\]  
Since $S_1\circ S_2\circ R\in \mathcal{F}_{\gK}(\gC_1,\gC_2)$, this shows that $P\notin \mathcal{F}_{\gK}(\gC_1,\gC_2)^\circ$ and thereby finishes the proof.
\end{proof}

\begin{thm}[Entanglement annihilating maps break some entanglement]
Let $\gC_1\subset V_1$ and $\gC_2\subset V_2$ denote proper cones and $\gK\subset V_3$ a resilient cone. If a positive map $P:V_1\ra V_2$ is $(\gC_1,\gC_2)$-entanglement annihilating, then
\[
P^{\otimes n} \in \mathcal{EB}_\gK(\gC^{\otimes_{\max} n}_1,\gC^{\otimes_{\min} n}_2)
\]
for every $n\in\N$.
\end{thm}

\begin{proof}
Assume that $P:V_1\ra V_2$ is $(\gC_1,\gC_2)$-entanglement annihilating and that there is a $n\in\N$ such that 
\[
P^{\otimes n} \notin \mathcal{EB}_\gK(\gC^{\otimes_{\max} n}_1,\gC^{\otimes_{\min} n}_2) .
\]
By definition there exist a $(\gK,\gC^{\otimes_{\max} n}_1)$-positive map $R:V_3\ra V^{\otimes n}_1$ and a $(\gC^{\otimes_{\min} n}_2,\gK)$-positive map $S:V^{\otimes n}_2\ra V_3$ such that $Q=S\circ P^{\otimes n}\circ R$ is not $\gK$-entanglement breaking. By Lemma \ref{lem:EntAnnUnderTrans} we find that $Q$ is $\gK$-entanglement annihilating, contradicting that $\gK$ was resilient. This finishes the proof.
\end{proof}

The following corollary follows from the previous theorem and the fact that $\mathcal{EB}_{\R^+}\lb \gC_1,\gC_2\rb$ is the set of $(\gC_1,\gC_2)$-positive maps.

\begin{cor}[Entanglement annihilating maps break Lorentz-entanglement]
Let $\gC_1\subset V_1$ and $\gC_2\subset V_2$ denote proper cones. A positive map $P:V_1\ra V_2$ is $(\gC_1,\gC_2)$-entanglement annihilating if and only if
\[
P^{\otimes n} \in \mathcal{EB}_{\gL_k}(\gC^{\otimes_{\max} n}_1,\gC^{\otimes_{\min} n}_2)
\]
for every $n,k\in\N$.
\label{cor:CharactEAbyLorentzEB}
\end{cor}

The previous corollary provides constraints on the set of entanglement annihilating maps. In the next section, we will take a closer look at these constraints in the case of $\gC_1=\gC_2=\PSD(\C^d)$ for some $d\in\N$.

\subsection{Positive maps on $\PSD(\C^d)$ factoring through Lorentz cones}
\label{sec:PosMapsFacThroughLor}

In the previous section, we saw that entanglement annihilating maps break entanglement with resilient cones, and in particular with any Lorentz cone (see Theorem \ref{thm:Main1}). An easy consequence of Lemma \ref{lem:KEBEasyProperties} and Corollary \ref{cor:CharactEAbyLorentzEB} is the following theorem giving criteria to detect maps that are not entanglement annihilating. In a certain sense this generalizes the reduction criterion from entanglement distillation~\cite{horodecki1999reduction}. 

\begin{thm}[Generalized reduction criterion]\label{thm:genRed}
Consider a linear map $Q:\M_{d_B}\ra \M_{d_A}$ factoring through a Lorentz cone $\gL_k$. For any entanglement annihilating map $P:\M_{d_A}\ra \M_{d_B}$ the composition $Q\circ P:\M_{d_A}\ra \M_{d_A}$ is entanglement breaking.
\end{thm}

Motivated by Theorem \ref{thm:genRed}, we will present a criterion for positive maps to factor through a Lorentz cone. For this we identify elements $X\in \R^{n+1}\otimes \M_d^{sa}$ with matrix-valued vectors $(X_0, X_1\ldots ,X_{n})$. The following proposition characterizes the maximal tensor product of the Lorentz cones and the positive semidefinite matrices: 

\begin{prop}[Maximal tensor product with Lorentz cone]
For
\[
X=(X_0,X_1\ldots ,X_{n})\in \R^{n+1}\otimes \M_d^{sa}
\]
the following are equivalent: 
\begin{enumerate}
\item $X\in \gL_n\otimes_{\max} \PSD(\C^d)$,
\item $X_0\geq 0$ 
and
\[
X_0\otimes X_0 - \sum^{n}_{s=1} X_s\otimes X_s\in \PSD(\C^d)\otimes_{\max} \PSD(\C^d) .
\]
\end{enumerate}
\label{prop:BasicLorentzPSDMax}
\end{prop}

\begin{proof}
We have $X\in \gL_n\otimes_{\max} \PSD(\C^d)$ if and only if
\begin{equation}
(\Tr\lbr YX_0\rbr,\Tr\lbr YX_1\rbr, \ldots ,\Tr\lbr YX_{n}\rbr)\in \gL_n
\label{equ:LorentzVect}
\end{equation}
for any $Y\geq 0$. For any $X\in \R^{n+1}\otimes \M_d^{sa}$ satisfying the second statement, we have $\Tr(YX_0)\geq 0$ and $\Tr(YX_s)\in\R$ for all $s\in \lset 1,\ldots ,n\rset$ and
\[
\Tr\lbr (Y\otimes Y)(X_0\otimes X_0 - \sum^{n}_{s=1} X_s\otimes X_s)\rbr = \Tr\lbr YX_0\rbr^2 -\sum^{n}_{s=1} \Tr\lbr YX_s\rbr^2\geq 0,
\]
which shows that \eqref{equ:LorentzVect} holds. This implies the first statement. Conversely, assume that \eqref{equ:LorentzVect} holds for any $Y\geq 0$. Clearly, this implies that $X_0\geq 0$. By the symmetries of the Lorentz cone we also have
\[
(\Tr\lbr YX_0\rbr, -\Tr\lbr YX_1\rbr, \ldots ,-\Tr\lbr YX_{n}\rbr)\in \gL_n,
\]
for any $Y\geq 0$. Using that the Lorentz cones are self-dual we find that 
\begin{align*}
\Tr&\lbr (Y\otimes Z)(X_0\otimes X_0 - \sum^{n}_{s=1} X_s\otimes X_s)\rbr \\
&= \Tr\lbr Y X_0\rbr\Tr\lbr Z X_0\rbr - \sum^{n}_{s=1} \Tr\lbr Y X_s\rbr\Tr\lbr Z X_s\rbr \geq 0
\end{align*}
for any $Y,Z\geq 0$, which implies the second statement.
\end{proof}

Consider a positive map $P:\M_d\ra\M_d$ with $\text{rk}(P)=n+1$ and satisfying $P=\vartheta_d\circ P^*\circ \vartheta_d$, or equivalently that $P\circ \vartheta_d$ is selfadjoint. Since $P\circ \vartheta_d$ is positive its spectral radius $R\lb P\circ \vartheta_d\rb$ is an eigenvalue and the corresponding eigenvector is positive semidefinite (see for example~\cite[Theorem 6.5]{wolf2012lecture}). Without loss of generality we can restrict to maps with spectral radius $1$ and, in this case, we have
\[
C_P = X_0\otimes X_0 + \sum^{n}_{i=1} \lambda_i X_i\otimes X_i ,
\]
with $X_0\geq 0$, $\lambda_i\in \lbr -1,1\rbr\setminus\lset 0\rset$, and matrices $X_i$ which are Hermitian and orthonormal. In the following, we will call this the \emph{canonical form} corresponding to $P$. We have the following:

\begin{thm}
Consider $P:\M_d\ra\M_d$ positive with $\text{rk}(P)=n+1$ and satisfying $P=\vartheta_d\circ P^*\circ \vartheta_d$. Let 
\[
C_P = X_0\otimes X_0 + \sum^{n}_{i=1} \lambda_i X_i\otimes X_i
\] 
be the canonical form, i.e.~$\lset X_i\rset^{n}_{i=0}$ forms an orthonormal basis of Hermitian matrices, $X_0\geq 0$, and $\lambda_i\in\lbr -1,1\rbr\setminus\lset 0\rset$. Then, the following are equivalent:
\begin{enumerate}
\item For all $i\in\lset 1,\ldots , n\rset$ we have $\lambda_i<0$.
\item There exist Hermitian matrices $Y_0,Y_1,\ldots , Y_{k}$ for some $k\in\N$ and
\[
\mu_0,\mu_1,\ldots ,\mu_{k}\geq 0,
\]
such that 
\[
C_P = \mu_0 Y_0\otimes Y_0 - \sum^{k}_{i=1} \mu_i Y_i\otimes Y_i .
\]
\end{enumerate}
\label{thm:AdvLorentzPSDMax}
\end{thm}

\begin{proof}
Clearly, the first statement implies the second. For the other direction, we define the matrix $M\in\R^{(k+1)\times (n+1)}$ with entries $M_{ij}=\Tr\lbr Y_i X_j\rbr$. The second statement implies that 
\begin{equation}
\text{diag}\lb 1,\lambda_1,\ldots , \lambda_{n}\rb = M^T \text{diag}\lb 1,-\mu_1,\ldots , -\mu_{k}\rb M.
\label{equ:defMmu}
\end{equation}
Since its left-hand-side has rank $n+1$, \eqref{equ:defMmu} can only hold when $k\geq n$. When $k=n$ either the first statement holds, or Sylvester's law of inertia implies that $M$ is singular. However, the latter would contradict the fact that $\text{diag}\lb 1,\lambda_1,\ldots , \lambda_{n}\rb$ is full-rank. 

Consider the case where $k>n$. By the singular value decomposition, we have $M=USV$ for orthogonal matrices $U\in \M_{k+1}\lb \R\rb$ and $V\in \M_{n+1}\lb \R\rb$ and a matrix $S\in \R^{(k+1)\times (n+1)}$ of the form $(S_1,\textbf{0})^T$, where $S_1\in\M_{n+1}\lb \R\rb$ is a positive diagonal matrix and $\textbf{0}\in \R^{(k-n)\times (n+1)}$ is the zero matrix. Setting $A = V\text{diag}\lb 1,\lambda_1,\ldots , \lambda_{n}\rb V^T$ and $B=U^T\text{diag}\lb 1,-\mu_1,\ldots , -\mu_{k}\rb U$ we find that \eqref{equ:defMmu} is equivalent to 
\[
A = S^TBS = S_1B_1 S_1,
\] 
where $B_1\in \M_{n+1}\lb\R\rb$ is the block in the block-decomposition 
\[
B = \begin{pmatrix} B_1 & B_2 \\ B_2^T & B_3\end{pmatrix}.
\]
By Cauchy's interlacing theorem, we have $\lambda^{\downarrow}_j\lb B\rb\geq\lambda^{\downarrow}_j\lb B_1\rb$ for all $j\in\lset 0,1,\ldots ,n\rset$, where $\lambda^{\downarrow}_j\lb \cdot\rb$ denotes the $j$th eigenvalue in decreasing order. Since the eigenvalues of $B$ are $(\mu_0,-\mu_1,\ldots ,-\mu_{n})$, we find that $\lambda^{\downarrow}_j\lb B_1\rb<0$ for all $j\in\lset 1,\ldots ,n\rset$. Because the eigenvalues of $A$ are $(1,\lambda_1,\ldots , \lambda_{n})$, either the first statement holds, or by Sylvester's law of inertia $S_1$ is singular contradicting the fact that $A$ is full-rank.
\end{proof}

We have the following corollary characterizing a subset of positive maps factoring through Lorentz cones.

\begin{cor}
For a selfadjoint positive map $P:\M_d\ra\M_d$ with $\text{rk}(P)=k+1$ the following are equivalent:
\begin{enumerate}
\item The spectral radius $R\lb P\rb$ is a simple eigenvalue and all other eigenvalues of $P$ are zero or negative.
\item There exists an $(\gL_k,\PSD_{d})$-positive map $\alpha:\R^{k+1}\ra \M_d$ such that
\[
P = \alpha\circ A \circ \alpha^*\circ \vartheta_d ,
\]
where $A:\R^{k+1}\ra \R^{k+1}$ denotes the $\gL_k$-positive map given by $A(e_0)=e_0$ and $A(e_i)=-e_i$ for each $i\in\lset 1,\ldots ,k\rset$.
\end{enumerate}
In either case, the positive map $P:\M_d\ra\M_d$ factors through $\gL_k$.
\label{cor:BasicFactorization}
\end{cor}

\begin{proof}
If $\mu_0 = R\lb P\rb$ is a simple eigenvalue and all other eigenvalues of $P$ are zero or negative, then we can decompose
\[
C_{P\circ \vartheta_d} = \mu_0 Y_0\otimes Y_0 - \sum^{k}_{i=1} \mu_i Y_i\otimes Y_i ,
\]
for Hermitian matrices $Y_0,Y_1,\ldots , Y_{k}$ and $\mu_0,\ldots ,\mu_{k}\geq 0$. By Proposition \ref{prop:BasicLorentzPSDMax} we may set $\alpha(e_i)=\mu^{1/2}_iY_i$ and, by linear extension, we obtain an $(\gL_k,\PSD_{d})$-positive map $\alpha:\R^{k+1}\ra \M_d$. Clearly, the second statement holds for this map. 

To show the reverse direction assume that there exists an $(\gL_k,\PSD_{d})$-positive map $\alpha:\R^{k+1}\ra \M_d$ such that
\[
P = \alpha\circ A \circ \alpha^*\circ \vartheta_d .
\]
By defining $Y_i:=\alpha(e_i)$ for any $i\in\lset 1,\ldots ,k\rset$ we have
\[
C_{P\circ \vartheta_d} = Y_0\otimes Y_0 - \sum^{k}_{i=1} Y_i\otimes Y_i ,
\]
and the first statement follows from Theorem \ref{thm:AdvLorentzPSDMax}.
\end{proof}

We will finish this section with a list of examples of well-known positive maps on the positive-semidefinite cones that factor through Lorentz cone. 
\begin{itemize}
\item \textbf{Reduction map:} The reduction map $R:\M_d\ra \M_d$ is given by \eqref{equ:Red} and was introduced in~\cite{horodecki1999reduction}. It is easy to check that $R\circ \vartheta_d$ is selfadjoint and has spectrum $\lset d-1,-1\rset$ with $d-1$ being a simple eigenvalue. By Corollary \ref{cor:BasicFactorization} we conclude that $R$ factors through a Lorentz cone, but we can do even better: It is known that $\text{SN}\lb C_{R\circ \vartheta_d}\rb = 2$ (see~\cite{terhal2000schmidt} for this fact and the definition of the Schmidt number $\text{SN}$), which implies that $R$ factors through $\PSD\lb \C^2\rb\simeq \gL_3$.   
\item \textbf{Breuer--Hall map:} The Breuer--Hall map $B:\M_4\ra\M_4$ was introduced in~\cite{breuer2006optimal,hall2006new} as an example of a non-decomposable positive map. It is given by 
\[
B(X) = \Tr(X)\one_4 - X - UX^TU^\dagger,
\]
for the antisymmetric unitary $U=\sigma_y\otimes \one_2$. It is easy to check that $B\circ \vartheta_4$ is selfadjoint and has spectrum $\lset 2,-2,0\rset$ where $2$ is a simple eigenvalue. By Corollary \ref{cor:BasicFactorization} we conclude that $B$ factors through a Lorentz cone. Specifically, it can be checked that 
\[
B = \alpha\circ A \circ \alpha^*\circ \vartheta_4 ,
\]
for the linear map $\alpha:\R^6\ra \M_4$ embedding $\gL_5\simeq \PSD(\qua^2)$ into $\PSD(\C^4)$ as
\[
\alpha(x) = \begin{pmatrix} x_0+x_5 & x_4-ix_3 & 0 & x_2-ix_1 \\ x_4+ix_3 & x_0-x_5 & -x_2+ix_1 & 0 \\ 0 & -x_2 - ix_1 & x_0 + x_5 & x_4 + ix_3 \\ x_2+ix_1 & 0 & x_4 -ix_3 & x_0 - x_5\end{pmatrix}.
\]
Hence, the Breuer--Hall map $B$ factors through the Lorentz cone $\gL_5$.
\item \textbf{Projections onto spin factors:} A set $\lset s_1,\ldots ,s_k\rset\subset \M_{d}$ with $k\geq 2$ is called a \emph{spin system} if each $s_i$ is a Hermitian unitary and such that $s_is_j + s_js_i=0$ when $i\neq j$. The subalgebra $A=\text{span}\lset \one_d, s_1,\ldots ,s_k\rset$ of Hermitian matrices is called a \emph{spin factor}. For each spin factor $A\subset \M^{sa}_d$ there is a positive faithful projection $P_A:\M_d\ra \M_d$ such that $P_A(\M^{sa}_d)= A$, and it is known that $P_A$ is atomic (and in particular not decomposable) if the spin factor $A$ is irreversible (see~\cite{stormer1980decomposition,stormer2012positive}). For a spin system $\lset s_1,\ldots ,s_k\rset\subset \M_{d}$ and the corresponding spin factor $A$ it is easy to check that 
\[
P_A = \phi\circ \phi^*,
\]
for the linear map $\phi:\R^{k+1}\ra \M_d$ given by $\phi(e_0)=\one_d$ and $\phi(e_i)=s_i$ for $i\in\lset 1,\ldots ,k\rset$. Moreover, note that for every $x\in \R^k$ we have 
\[
\lb \sum^k_{i=1} x_i s_i\rb^2 = \|x\|^2_2 \one_d,
\] 
and therefore $\phi(\gL_k)\subseteq \PSD(\C^d)$. We conclude that $P_A$ factors through the Lorentz cone $\gL_k$.
\end{itemize}

It should be emphasized that the previous examples contain many positive maps that are non-decomposable. By the duality between decomposable positive maps and completely positive maps that are completely copositive (see~\cite{stormer1982decomposable}), and Theorem \ref{thm:genRed} we find many examples of completely positive maps that are completely copositive, but not entanglement annihilating. In particular, this shows that Proposition \ref{prop:StrongPPT2} does not generalize to all tensor powers. 

\section{Partial results for cones with a symmetric base}\label{sec:ConeSymBase}

In the following, we will focus on cones $\gC_X$ associated to a finite-dimensional normed space $X$. By Theorem \ref{thm:resilienceVstildePalpha}, resilience of $\gC_X$ can be decided by showing that every entanglement annihilating central map is entanglement breaking. Here, we will use the theory of Banach space tensor norms to obtain partial results aiming at a characterization of when central maps are entanglement annihilating.   

\subsection{The tensor radii of normed spaces and maps}\label{sec:SummaryPaperB}

Two natural tensor norms can be defined on the algebraic tensor product $X^{\otimes k}$ of a finite-dimensional, real, normed space $X$: the \emph{injective tensor norm},  given for $z\in X^{\otimes k}$ by
\[
\|z\|_{\e_k(X)} = \sup\Big\lset |(\lambda_1\otimes \cdots \otimes \lambda_k)(z)| ~:~\lambda_1,\ldots ,\lambda_k\in B_{X^*}\Big\rset,
\]
and the \emph{projective tensor norm}, given by
\[
\|z\|_{\pi_k(X)} = \inf\Big\lset \sum^n_{i=1} \|x^{(1)}_i\|_X\cdots \|x^{(k)}_i\|_X ~:~n\in\N,\ z=\sum^n_{i=1} x^{(1)}_{i}\otimes\ldots \otimes x^{(k)}_{i}\Big\rset .
\]
It is well-known that $\|z\|_{\e_k(X)}\leq \|z\|_{\pi_k(X)}$ for all $k\in\N$, which implies the inclusion 
\[
\gC_{\pi_k(X)}\subseteq \gC_{\varepsilon_k(X)},
\]
of the corresponding cones. In~\cite{paperB}, we studied the quantities
\begin{equation} \label{eq:def-tau} \tau_k(T) = \| T^{\otimes k} \|^{1/k}_{\e_k(X) \to \pi_k(Y)} = \lb \sup_{z\in X^{\otimes k}}\frac{\|T^{\otimes k}z\|_{\pi_k(Y)}}{\|z\|_{\e_k(X)}}\rb^{\frac{1}{k}},
\end{equation}
for every $k\in\N$ and any linear operator $T:X\ra Y$. Furthermore, we showed that the limit $\tau_\infty(T):=\lim_{k\ra\infty} \tau_k(T)$ exists and satisfies 
\begin{equation} \label{eq:tau-inequalities}\|T\|_{X \to Y} = \tau_1(T) \leq \tau_k(T) \leq \tau_{\infty}(T) \leq \|T\|_{N(X \to Y)}.\end{equation}
The quantity $\tau_\infty(T)$ is called the \emph{tensor radius of $T$} and in the special case of $X=Y$ and $T=\ident_X$ we call $\rho_\infty(X):=\tau_\infty(\ident_X)$ the \emph{tensor radius of the normed space $X$}. The following theorem collects the main results of~\cite{paperB}:

\begin{thm}[\cite{paperB}]
Let $X$ be a $n$-dimensional normed space.
\begin{itemize}
\item We have 
\[
\sqrt{n} \leq \rho_\infty(X) \leq n,
\]
with $\rho_\infty(X) = \|\ident_X\|_{N(X\ra X)}= n$ if and only if $X$ is Euclidean.
\item If $X$ has enough symmetries, then we have
\[\rho_{\infty}(X) = \frac{n}{\mathrm{d}(X,\ell_2^d)} ,\]
where $\text{d}(X,Y)$ denotes the Banach--Mazur distance, given by
\[
\text{d}(X,Y)=\inf\lset \|U\|_{X\ra Y}\|U^{-1}\|_{Y\ra X}~:~U:X\ra Y\text{ linear bijection}\rset .
\]

\item If $X$ is Euclidean, then we have $\tau_{\infty}(T) = \|T\|_{N}$ for every linear operator $T : X \to Y$ or $T : Y \to X$, where $Y$ is an arbitrary finite-dimensional normed space.
\end{itemize}
\label{thm:MainResPaperB}
\end{thm}

In the next section, we will show how the tensor radii can be used to show a certain kind of partial-entanglement annihilation.

\subsection{Tensor products of cones from tensor norms}\label{sec:TensorProdFromNorms}

Let $X$ denote a finite-dimensional normed space. In the following, we set $e_0=(1,0_X)\in \R\oplus X$ and we sometimes identify $X$ with its embedding into $\R\oplus X$ by $x\mapsto (0,x)$. With this convention, we consider the subspace $X_k \subset \lb \R\oplus X\rb^{\otimes k}$ given by
\[
X_k = \text{span}\lb\lset e^{\otimes k}_0\rset\cup X^{\otimes k}\rb\subset (\R\oplus X)^{\otimes k} .
\]
Note that $X_k$ consists of the vectors of the form $e^{\otimes k}_0 + z$, where $z\in X^{\otimes k}$ is identified with its canonical embedding into $(\R\oplus X)^{\otimes k}$. We will denote by ${\Pi}_{X_k}:\lb \R\oplus X\rb^{\otimes k} \ra \lb \R\oplus X\rb^{\otimes k}$ the orthogonal projection onto the subspace $X_k$. The following lemma is a multpartite version of~\cite[Proposition 2.25]{lami2018non} (see also~\cite[Lemma S13]{aubrun2019universal}) and for convenience we include a proof in Appendix \ref{sec:ProofLemTPBanachVsCone}. 

\begin{lem}[Tensor products on normed spaces and cones]
\label{lem:TPBanachVsCone}
For a finite-dimensional normed space $X$ we have
\[
\gC^{\otimes_{\max} k}_X\cap X_k = {\Pi}_{X_k}(\gC^{\otimes_{\max} k}_X) = \gC_{\varepsilon_k(X)},
\]
\[
\gC^{\otimes_{\min} k}_X\cap X_k = {\Pi}_{X_k}(\gC^{\otimes_{\min} k}_X) = \gC_{\pi_k(X)}.
\]
\end{lem}

Inspired by Lemma \ref{lem:TPBanachVsCone} we can define tensor products of the cone $\gC_X$ that are closely related to the injective and projective tensor norms.  

\begin{defn}[The hat and check tensor product]
For any finite-dimensional normed space $X$, we define the closed cones
\[
\gC^{\check{\otimes} k}_X = \lset z\in \gC^{\otimes_{\max} k}_X ~:~\Pi_{X_k}(z)\in \gC_{\pi_k(X)}\rset,
\]
and 
\[
\gC^{\hat{\otimes} k}_X = \gC^{\otimes_{\min} k}_X \vee \gC_{\varepsilon_k(X)}.
\]
\end{defn}

To illustrate these tensor product, we compute them for the case of $k=2$ and $\gC_X = \gL_3 \simeq \PSD(\C^2)$, i.e., the case of two qubits. 

\begin{example}
The Pauli basis is an orthogonal basis of $\M_2$ with respect to the Hilbert-Schmidt inner product, and given by 
\[
\sigma_0 =\begin{pmatrix} 1 & 0 \\ 0 & 1\end{pmatrix},\, \sigma_1 =\begin{pmatrix} 0 & 1 \\ 1 & 0\end{pmatrix},\,\sigma_2 =\begin{pmatrix} 0 & -i \\ i & 0\end{pmatrix},\,\sigma_3 =\begin{pmatrix} 1 & 0 \\ 0 & -1\end{pmatrix}.
\]
The spinor map $S:\R^4\ra \M^{sa}_2$ is given by 
\[
S(t,x_1,x_2,x_3) = t \sigma_0 + \sum^3_{i=1} x_i \sigma_i ,
\]
and it defines an order-isomorphism between the cones $\gL_3$ and $\PSD(\C^2)$. It is easy to compute that 
\begin{equation}\label{equ:normsl22}
\|z\|_{\pi_2(\ell^n_2)} = \|z\|_1, \quad\text{ and }\quad \|z\|_{\varepsilon_2(\ell^n_2)} = \|z\|_\infty,
\end{equation}
where we interprete $z\in \R^n\otimes \R^n$ as an $n\times n$ matrix, and where $\|\cdot\|_1$ is the trace-norm, and $\|\cdot\|_\infty$ is the operator norm. Finally, we need the moment map $M:\M_2\otimes \M_2 \ra \M_4$ given by 
\[
M(X)_{ij} = \Tr\lbr X(\sigma_i\otimes \sigma_j)\rbr.
\]
Using the spinor map and \eqref{equ:normsl22}, it is easy to verify that 
\[
Z\in \PSD(\C^2)\hat{\otimes}\PSD(\C^2) \text{ if and only if } Z = S + W, 
\]
for some $S\in \PSD(\C^2)\otimes_{\min}\PSD(\C^2)$ and some $W\in \M_2\otimes \M_2$ satisfying $M(W)_{ij}=0$ whenever $i=0\neq j$ or $i\neq 0 =j$, and such that
\[
\Tr\lbr W\rbr \geq \|\lbr M(W)_{ij}\rbr^3_{i,j=1}\|_\infty .
\]
Again using the spinor map and \eqref{equ:normsl22}, we have
\[
\PSD(\C^2)\check{\otimes}\PSD(\C^2) = \lset Z\in \PSD(\C^2)\otimes_{\max}\PSD(\C^2):\Tr\lbr Z\rbr \geq \|\lbr M(Z)_{ij}\rbr^3_{i,j=1}\|_1\rset .
\]

\end{example}

It is easy to check that 
\[
\gC^{\otimes_{\min}k}_X \subseteq \gC^{\check{\otimes} k}_X \subseteq \gC^{\otimes_{\max}k}_X \quad\text{ and }\quad\gC^{\otimes_{\min}k}_X \subseteq \gC^{\hat{\otimes} k}_X \subseteq \gC^{\otimes_{\max}k}_X ,
\]
and in most cases these inclusions are strict. This shows that $\check{\otimes}$ and $\hat{\otimes}$ are tensor products in the category of convex cones. Moreover, these two tensor products are dual to each other:

\begin{lem}
For any normed space $X$ we have 
\[
\lb \gC^{\hat{\otimes} k}_X \rb^* = \gC^{\check{\otimes} k}_{X^*}.
\]
\end{lem}

\begin{proof}
Any $z\in  \gC^{\hat{\otimes} k}_X$ can be written as a sum $z=z_1 + z_2$ with $z_1\in \gC^{\otimes_{\min} k}_X$ and
\[
z_2\in \gC_{\varepsilon_k(X)} = \Pi_{X_k}\lb \gC^{\otimes_{\max} k}_X\rb,
\]
by Lemma \ref{lem:TPBanachVsCone}. For any $y\in \gC^{\check{\otimes} k}_{X^*}$ we have $\braket{y}{z_1}\geq 0$, and we have 
\[
\braket{y}{z_2} = \braket{\Pi_{X_k}(y)}{z_2}\geq 0,
\]
since $\Pi_{X_k}(y)\in \gC_{\pi_k(X^*)} = (\gC_{\varepsilon_k(X)})^*$. We conclude that $\gC^{\check{\otimes} k}_{X^*}\subseteq \lb \gC^{\hat{\otimes} k}_X\rb^*$.

For the converse direction, consider $y\in \lb C^{\hat{\otimes} k}_X\rb^*$. Since $\gC^{\otimes_{\min} k}_X \subseteq \gC^{\hat{\otimes} k}_X$ we have $y\in \gC^{\otimes_{\max} k}_{X^*}$, and since
\[
\gC_{\varepsilon_k(X)} = \Pi_{X_k}\lb \gC^{\otimes_{\max} k}_X\rb\subseteq \gC^{\hat{\otimes} k}_X,
\]
by Lemma \ref{lem:TPBanachVsCone}, we find that $\Pi_{X_k}\lb y\rb\in \lb \gC_{\varepsilon_k(X)}\rb^* = \gC_{\pi_k(X^*)}$. 
\end{proof}

The relevance of the tensor products $\check{\otimes}$ and $\hat{\otimes}$ comes from the following theorem, where, given $\alpha \in \R$ and $P : X \to Y$, we denote by $\tilde{P}_\alpha$ the central map $\alpha \oplus P$.

\begin{thm}
Let $X,Y$ be normed spaces, $P:X\ra Y$ a linear map, and $\alpha\in\R^+$. The following are equivalent:
\begin{enumerate}
\item $\tau_\infty\lb P\rb\leq \alpha$.
\item $\tilde{P}^{\otimes k}_\alpha\lb \gC^{\hat{\otimes} k}_X\rb \subseteq \gC^{\otimes_{\min} k}_Y$ for any $k\in\N$.
\item $\tilde{P}^{\otimes k}_\alpha\lb \gC^{\otimes_{\max} k}_X\rb \subseteq \gC^{\check{\otimes} k}_Y$ for any $k\in\N$.
\end{enumerate}
\label{thm:CheckHatEA}
\end{thm}

\begin{proof}
Note first, that any of the three conditions implies that $\tilde{P}_\alpha$ is $(\gC_X,\gC_Y)$-positive, or equivalently that $\|P\|\leq \alpha$. This implies that $\tilde{P}^{\otimes k}_\alpha$ is both $(\gC_X^{\tmin k},\gC_Y^{\tmin k})$-positive and $(\gC_X^{\tmax k},\gC_Y^{\tmax k})$-positive.

We first show that $(1) \Rightarrow (2)$. To show $(2)$, it suffices to show that 
$\tilde{P}_\alpha^{\otimes k}(w) \in \gC_Y^{\tmin k}$ whenever $w\in \gC_X^{\tmin k}$ or $w\in \gC_{\e_k(X)}$. Since $\tilde{P}_\alpha^{\otimes k}$ is $(\gC_X^{\tmin k},\gC_Y^{\tmin k})$-positive, the first case is clear. For the second case, consider $w \in \gC_{\e_k(X)}$, which can be written as $w=z_0e^{\otimes k}_0 + z$ for $z\in X^{\otimes k}$ such that $z_0\geq \|z\|_{\varepsilon_k(X)}$. Now, we have 
\[
\tilde{P}^{\otimes k}_\alpha\lb w \rb = \alpha^k z_0 e^{\otimes k}_0 + P^{\otimes k}(z) .\]
Since $\|P^{\otimes k}(z)\|_{\pi_k(Y)} \leq \tau_k(P)^k \|z\|_{\e_k(X)} \leq \alpha^k z_0$, we conclude that
\[
\tilde{P}^{\otimes k}_\alpha(w) \in \gC_{\pi_k(Y)} \subset \gC_Y^{\tmin k},
\]
as needed.

To show that $(2) \Rightarrow (3)$ consider $z \in \gC^{\otimes_{\max} k}_X$. We already observed that $\tilde{P}^{\otimes k}(z) \in \gC^{\otimes_{\max} k}_X$. By Lemma \ref{lem:TPBanachVsCone}, $\Pi_{X_k}(z)\in \gC_{\varepsilon_k(X)} \subset \gC_X^{\hat{\otimes} k}$; and (2) implies that
\[
\Pi_{Y_k}\lb \tilde{P}^{\otimes k}_\alpha\lb z\rb\rb = \tilde{P}^{\otimes k}_\alpha\lb \Pi_{X_k}\lb z\rb\rb \in \gC^{\otimes_{\min} k}_Y,
\]
so that $(3)$ holds by definition of $\check{\otimes}$. 

We finally show that $(3) \Rightarrow (1)$. Consider $z \in X^{\otimes k}$ such that $\|z\|_{\e_k(X)} \leq 1$, and set $w =e_0^{\otimes k} + z \in \gC_{\e_k(X)} \subset \gC_X^{\tmax k}$. By assumption, we have 
\[
\Pi_{Y_k}\lb\tilde{P}^{\otimes k}_\alpha\lb w \rb \rb= \alpha^k e^{\otimes k}_0 + P^{\otimes k}(z) \in \gC_{\pi_k(Y)},
\] 
and we conclude that $\|P^{\otimes k}(z)\|_{\pi_k(Y)}\leq \alpha^k$. Since $z$ and $k\in\N$ was arbitrary we find that $(1)$ holds. 
\end{proof}

By combining Theorem \ref{thm:MainResPaperB} and Theorem \ref{thm:CheckHatEA} we can illuminate the limits of the proof-technique from Section \ref{sec:lorentz}: By Theorem \ref{thm:resilienceVstildePalpha} resilience of a cone $\gC_X$ can be decided by focusing on central maps, of the form $\tilde{P}_\alpha$ for $P:X \to Y$. If the entanglement of tensors in $\gC^{\hat{\otimes} k}_X$ (as used in Section \ref{sec:lorentz}) is annihilated we can only conclude that $\tau_\infty\lb P\rb\leq \alpha$. Except in the cases where $X$ is Euclidean (and $\gC_X$ is a Lorentz cone) we have $\tau_\infty\lb P\rb<\| P\|_{N(X\ra X)}$, and we cannot conclude that $\tilde{P}_\alpha$ is entanglement breaking.  

We conclude this section with three corollaries that follow from Theorem \ref{thm:resilienceVstildePalpha}, Theorem \ref{thm:MainResPaperB} and Theorem \ref{thm:CheckHatEA}:

\begin{cor}\label{cor:First} 
If $\gC_X$ is not resilient, then we have  
\[
\tau_{\infty}(P)<\|P\|_{N(X\ra X)},
\] 
for some linear map $P:X\ra X$. In particular, this implies that $X$ is not Euclidean.
\end{cor}

Corollary \ref{cor:First} provides an alternative way to prove Theorem \ref{thm:Main1} by using the results from \cite{paperB}.  

\begin{cor}
If $X$ is not Euclidean, then there exists a linear map $P:X\ra X$ and an $\alpha\geq 0$ such that $\tilde{P}_\alpha$ is not entanglement breaking and such that 
\[
\tilde{P}^{\otimes k}_\alpha\lb \gC^{\hat{\otimes} k}_X\rb \subseteq \gC^{\otimes_{\min} k}_X \quad\text{ and }\quad\tilde{P}^{\otimes k}_\alpha\lb \gC^{\otimes_{\max} k}_X\rb \subseteq \gC^{\check{\otimes} k}_X
\]  
for any $k\in\N$.
\label{cor:GeneralPartialEA}
\end{cor}

The third corollary, considers the case of spaces with enough symmetries:

\begin{cor}
If $X$ has enough symmetries and
\[
\frac{d}{\mathrm{d}(X,\ell_2^d)}\leq \alpha < \|\ident_X\|_{N(X\ra X)} = d,
\]
then the isotropic map $I_{\alpha}:=I_{\alpha,1}$ is not entanglement breaking and satisfies
\begin{enumerate}
\item $I_{\alpha}^{\otimes k}\lb \gC^{\hat{\otimes} k}_X\rb \subseteq \gC^{\otimes_{\min} k}_X$ for any $k\in\N$.
\item $I_{\alpha}^{\otimes k}\lb \gC^{\otimes_{\max} k}_X\rb \subseteq \gC^{\check{\otimes} k}_X$ for any $k\in\N$. 
\end{enumerate}
\label{cor:IsotropicPartialEA}
\end{cor}

It should be noted that Corollary \ref{cor:GeneralPartialEA} identifies many convex cones and natural tensor products (i.e., $\hat{\otimes}$ and $\check{\otimes}$) for which non-trivial tensor-stable positive maps exist. To our knowledge this is the first result of this kind, where arbitrary tensor-powers of a map can be controlled in a non-trivial setting. 

\subsection{A candidate for resilience?}\label{sec:CandidateForResilience}

Inspired by the results from the previous section, we will take a closer look at the special case of $X=\ell^d_1$. It is well-known that $\ell^d_1$ has enough symmetries and that $\mathrm{d}(X,\ell_2^d) = \sqrt{d}$. Corollary \ref{cor:RedToIso} implies that the pair $(\gC_{\ell^d_1},\gC_{\ell^d_1})$ is resilient if and only if there is an isotropic map $I_{\alpha,\beta}$ (see \eqref{equ:IsoMap}) that is entanglement annihilating and not entanglement breaking. After choosing $\beta=1$ without loosing generality, we conclude by Corollary \ref{cor:IsotropicPartialEA} that $I_\alpha = I_{\alpha,1}$ could only be entanglement annihilating without being entanglement breaking for $\sqrt{d} \leq \alpha \leq d$. We state this as a proposition: 

\begin{prop}
The cone $\gC_{\ell^d_1}$ is resilient if and only if there exists a
\[
\alpha \in \lbr \sqrt{d}, d\rb,
\]
such that the isotropic map $I_{\alpha}$ is entanglement annihilating.
\end{prop} 

We can state the most simple open problem in this direction: 

\begin{que}
Is $I_{\alpha}:\R^3\ra \R^3$ entanglement annihilating for $\alpha=\sqrt{2}$? 
\label{que:Easyl1}
\end{que} 

It is easy to show that $I^{\otimes 2}_{\sqrt{2}}\lb \gC^{\otimes_{\max} 2}_{\ell^2_1}\rb \subset \gC^{\otimes_{\min} 2}_{\ell^2_1}$. Surprisingly, numerical experiments show that 
\[
I^{\otimes 3}_{\sqrt{2}}\lb \gC^{\otimes_{\max} 3}_{\ell^2_1}\rb \subset \gC^{\otimes_{\min} 3}_{\ell^2_1}.
\]
To further explore whether Question \ref{que:Easyl1} is reasonable, we can use Corollary \ref{cor:CharactEAbyLorentzEB}: If $I_{\sqrt{2}}$ were entanglement annihilating, then it would break entanglement with any Lorentz cone $\gL_k$. This holds as well, and we even have the following more general result for \emph{symmetric cones}, i.e., closed convex cones $\gC$ satisfying $\gC=\gC^*$ and such that their automorphism group acts transitively on their interior.

\begin{thm}\label{thm:BreakingEntWithSymmetricCone}
Let $\gC\subset V$ denote a symmetric cone in a Euclidean space $V$. For any $k\in \N$ we have 
\[
(\ident_{V}\otimes I_{\sqrt{k}})(\gC \otimes_{\max} \gC_{\ell^k_1} )\subseteq \gC \otimes_{\min} \gC_{\ell^k_1}  .
\]
\end{thm}

We will give the proof of the previous theorem and relevant background on the theory of symmetric cones in Appendix \ref{sec:Jordan}. The family of symmetric cones contains the Lorentz cones $\gL_n$, the positive semidefinite cones $\PSD\lb \R^d\rb$, $\PSD\lb \C^d\rb$, and $\PSD\lb \qua^d\rb$, and the cone $\PSD\lb \oc^3\rb$ of positive semidefinite $3\times 3$ matrices with octonionic entries. Theorem \ref{thm:BreakingEntWithSymmetricCone} can be seen as a generalization of~\cite[Theorem 6.6]{passer2018minimal} on inclusion constants of matrix convex sets. The minimal matrix convex set $\mathcal{W}^{\min}_n(K)$ and the maximal matrix convex set $\mathcal{W}^{\max}_n(K)$ are exactly the minimal and maximal tensor products, respectively, of a cone over the convex base set $K\subseteq \R^d$ with $\PSD(\C^n)$. Theorem \ref{thm:BreakingEntWithSymmetricCone} shows that in the case of $K=B_{\ell^k_1}$, the positive semidefinite cones $\PSD(\C^n)$ in these definitions can be replaced by any symmetric cone without changing the inclusion constant. Moreover, even in the original case of matrix convex sets our proof seems to be simpler than the one given in~\cite{passer2018minimal}.

\section*{Acknowledgments}
This work was supported in part by ANR (France) under the grant ESQuisses (ANR-20-CE47-0014-01). We acknowledge funding from the European Union’s Horizon 2020 research and innovation programme under the Marie Sk\l odowska-Curie Action TIPTOP (grant no. 843414).

\appendix

\section{Properties of linear maps under finite tensor powers}
\label{app:ProofOfPropkEA}

Given cones $\gC_1, \gC_2$, we denote by $\mathcal{EB}(\gC_1,\gC_2)$ the cone of $(\gC_1,\gC_2)$-entanglement breaking maps. We start with a definition of certain sequences of convex cones of linear maps:

\begin{defn}[Compatible sequences of mapping cones]
Let\label{def:mapping-cones} $\gC_1\subset V_1$ and $\gC_2\subset V_2$ be proper cones. For every $k\in \N$ let $\mathcal{T}_k$ be a closed convex cone of linear maps from $V^{\otimes k}_1$ into $V^{\otimes k}_2$. We call the sequence $\lb \mathcal{T}_k\rb_{k\in\N}$ a \emph{$(\gC_1,\gC_2)$-compatible sequence of mapping cones} if the following conditions are satisfied:
\begin{enumerate}
\item We have $\mathcal{EB}\lb \gC_1, \gC_2\rb\subseteq \mathcal{T}_1$.
\item For every $k\in\N$ and every $Q\in \mathcal{T}_{k}$ we have $x\alpha\otimes Q\in \mathcal{T}_{k+1}$ and $Q\otimes x\alpha\in \mathcal{T}_{k+1}$ for any $\alpha\in \gC^*_1$ and $x\in \gC_2$.
\end{enumerate}
\end{defn}

Note that combining the two properties in the previous definition shows 
\[
\mathcal{EB}\lb \gC^{\otimes_{\max} k}_1, \gC^{\otimes_{\min} k}_2\rb \subseteq \mathcal{T}_k
\]
for every $k\in\N$. Examples of compatible sequences of mapping cones include the mapping cones in the sense of~\cite{skowronek2011cones}, e.g., the $n$-positive maps, the decomposable maps, and the entanglement breaking maps on Hilbertian tensor powers of the positive semidefinite cones, but also the set of \emph{$k$-entanglement annihilating maps} $P:V_1\ra V_2$ satisfying
\[
P^{\otimes k}\lb \gC^{\otimes_{\max} k}_1\rb\subseteq \gC^{\otimes_{\min}k}_2 ,
\]
for some pair of proper cones $\gC_1$ and $\gC_2$. 

Denote by $L(V_1,V_2)$ the space of linear maps between vector spaces $V_1$ and $V_2$. If $\mathcal{T} \subset L(V_1,V_2)$ is a proper cone, we may define via trace duality 
\[ \mathcal{T}^\sharp = \{ Q \in L(V_2,V_1) \, : \, \Tr [QP] \geq 0 \ \forall P \in \mathcal{T} \}. \]
Note that $\mathcal{T}^\sharp$ is isomorphic to the dual cone $\mathcal{T}^*$ via the canonical isomorpism between $L(V_2,V_1)$ and $L(V_1,V_2)^*$. In particular, the equation $(\mathcal{T}^\sharp)^\sharp=\mathcal{T}$ is an instance of the bipolar theorem. We now describe the conditions which are dual to Definition \ref{def:mapping-cones}. They are obtained by arguing as in the proof of Lemma \ref{lem:dualityEB}.

\begin{lem}\label{lem:TraceDualityForMappingCChain}
Let $\gC_1\subset V_1$ and $\gC_2\subset V_2$ be proper cones, and $\lb \mathcal{T}_k\rb_{k\in\N}$ a $(\gC_1,\gC_2)$-compatible sequence of mapping cones. 
Then, we have:
\begin{enumerate}
\item For every $k\in\N$ we have $\mathcal{T}_k^\sharp \subseteq \mathcal{P}\lb \gC^{\otimes_{\min} k}_2,\gC^{\otimes_{\max} k}_1\rb$.
\item For every $k\in \N$ and every $Q\in \mathcal{T}_{k+1}^\sharp$, we have
\[
\lb  \alpha\otimes \ident^{\otimes k}_{V_1}\rb \circ Q\circ \lb x\otimes \ident^{\otimes k}_{V_2}\rb \in \mathcal{T}_{k}^\sharp
\]
and 
\[
\lb  \ident^{\otimes k}_{V_1}\otimes \alpha\rb \circ Q\circ \lb \ident^{\otimes k}_{V_2}\otimes x\rb \in \mathcal{T}_{k}^\sharp,
\]
for any $\alpha\in \gC^*_1$ and $x\in \gC_2$.
\end{enumerate}
\end{lem} 

We will prove the following theorem:

\begin{thm}\label{thm:FullyGenkEA}
Let $\gC_1\subset V_1$ and $\gC_2\subset V_2$ be proper cones and $\lb \mathcal{T}_k\rb_{k\in\N}$ a $(\gC_1,\gC_2)$-compatible sequence of mapping cones. Then, any $P\in \mathcal{EB}\lb \gC_1, \gC_2\rb \cap \inter\lb \mathcal{T}_1\rb$ satisfies $P^{\otimes k}\in \inter\lb \mathcal{T}_k\rb$ for every $k\in\N$.
\end{thm}

If we define 
\[
\widetilde{\mathcal{T}}_k = \lset P:V_1\ra V_2~:~P^{\otimes k}\in \mathcal{T}_k\rset, 
\]
then Theorem \ref{thm:FullyGenkEA} shows that 
\[
\mathcal{EB}\lb \gC_1, \gC_2\rb \cap \inter\lb \mathcal{T}_1\rb \subseteq \bigcap_{k\in\N} \widetilde{\mathcal{T}}_k.
\]
To prove Theorem \ref{thm:FullyGenkEA}, we fix $e^*_1\in \inter\lb \gC^*_1\rb$ and $e_2\in \inter\lb \gC_2\rb$. Next, we define 
\[
\mu_{k}\lb P\rb := \inf_{Q\in \Sigma_{k}}\Tr\lbr P\circ Q\rbr ,
\]
for any $P:V_1^{\otimes k}\ra V^{\otimes k}_2$, where
\[
\Sigma_{k} :=\Big\lset Q\in \mathcal{T}_k^\sharp \,:\, (e^*_1)^{\otimes k}\lb Q(e^{\otimes k}_2)\rb=1\Big\rset .
\] 
By the first property in Lemma \ref{lem:TraceDualityForMappingCChain}, we have $(e^*_1)^{\otimes k}\lb Q(e^{\otimes k}_2)\rb=0$ for a map $Q\in \mathcal{T}_k^\sharp$ if and only if $Q=0$. Therefore, we conclude that $\Sigma_k$ is compact for every $k\in\N$, and the infimum in the definition of $\mu_k$ is attained. Since $\mathcal{T}_k = (\mathcal{T}_k^\sharp)^\sharp$, we have that $\mu_{k}\lb P\rb\geq 0$ if and only if $P\in \mathcal{T}_k$, and $\mu_{k}\lb P\rb> 0$ if and only if $P\in \inter\lb\mathcal{T}_k\rb$. We will now show the following:

\begin{lem}\label{lem:Intermediate}
If $P:V_1\ra V_2$ is $(\gC_1,\gC_2)$-entanglement breaking, then we have
\[
\mu_{k}\lb P^{\otimes k}\rb \geq \mu_{1}\lb P\rb^k . 
\]
\end{lem}

\begin{proof}
We will show a slightly more general statement. Fixing $k\in \N$, consider any linear map $S:V_1^{\otimes (k-1)}\ra V_2^{\otimes (k-1)}$ and assume that $P=\sum^N_{i=1} b_i\alpha_i$ for $b_i\in \gC_2$ and $\alpha_i\in \gC_1^*$. We will show that 
\begin{equation}
\mu_{k}\lb P\otimes S\rb \geq \mu_{1}\lb P\rb \mu_{k-1}\lb S\rb.
\label{equ:Main}
\end{equation}
Using \eqref{equ:Main}, the proof is finished by induction after setting $S=P^{\otimes (k-1)}$. 

Consider a $Q\in \Sigma_{k}$ satisfying 
\[
\mu_{k}\lb P\otimes S\rb = \Tr\lbr (P\otimes S)\circ Q\rbr .
\]
Inserting the decomposition of $P$, we find
\begin{align*}
\Tr\lbr (P\otimes S)\circ Q\rbr &= \sum^N_{i=1} \Tr\lbr (b_i\alpha_i\otimes S)\circ Q\rbr \\
&= \sum^N_{i=1} \Tr\lbr S\circ \lb\alpha_i\otimes \ident^{\otimes (k-1)}_{V_1}\rb\circ Q(b_i\otimes \cdot)\rbr \\
&=\sum^N_{i=1} \Tr\lbr S\circ \tilde{Q}_i\rbr ,
\end{align*}
where we defined the linear maps
\[
\tilde{Q}_i := \lb\alpha_i\otimes \ident^{\otimes (k-1)}_{V_1}\rb\circ Q(b_i\otimes \cdot):V_2^{\otimes (k-1)}\ra V_1^{\otimes (k-1)} .
\]
By the second property of $\mathcal{T}_{k}^\sharp$ in Lemma \ref{lem:TraceDualityForMappingCChain}, we have $\tilde{Q}_i\in \mathcal{T}_{k-1}^\sharp$ for every $i\in \lset 1,\ldots ,N\rset$. Since 
\[
(e^*_1)^{\otimes (k-1)}\lb \tilde{Q}_i(e^{\otimes (k-1)}_2)\rb = 0,
\]
if and only if $\tilde{Q}_i=0$, we conclude that
\[
\mu_{k}\lb P\otimes S\rb \geq \sum^N_{i=1} (e^*_1)^{\otimes (k-1)}\lb \tilde{Q}_i(e^{\otimes (k-1)}_2)\rb \mu_{k-1}(S).
\]
Now, note that 
\[
\sum^N_{i=1} (e^*_1)^{\otimes (k-1)}\lb \tilde{Q}_i(e^{\otimes (k-1)}_2)\rb = \Tr\lbr P \circ Q'\rbr,
\]
for the linear map 
\[
Q' := \lb\ident_{V_1}\otimes (e^*_1)^{\otimes (k-1)}\rb \circ Q\lb \cdot \otimes e^{\otimes (k-1)}_2\rb.
\]
Using the second property of $\mathcal{T}_{k}^\sharp$ in Lemma \ref{lem:TraceDualityForMappingCChain} repeatedly, we find that $Q'\in \mathcal{T}_1^\sharp$, and it is easy to check that 
\[
e^*_2\lb Q'(e_1)\rb = (e^*_2)^{\otimes k}\lb Q(e^{\otimes k}_1)\rb = 1 .
\]
We conclude that $Q'\in \Sigma_{1}$ and 
\[
\sum^N_{i=1} (e^*_1)^{\otimes (k-1)}\lb \tilde{Q}_i(e^{\otimes (k-1)}_2)\rb \geq \mu_{1}(P).
\]
Finally, \eqref{equ:Main} follows by combining the estimates from above. 
\end{proof}

Now, we are ready to prove Theorem \ref{thm:FullyGenkEA}:

\begin{proof}[Proof of Theorem \ref{thm:FullyGenkEA}]
Consider $P\in \mathcal{EB}\lb \gC_1, \gC_2\rb \cap \inter\lb \mathcal{T}_1\rb$ and note that $\mu_1(P)>0$. For every $k\in \N$ we can apply Lemma \ref{lem:Intermediate} to show that
\[
\mu_k\lb P^{\otimes k}\rb \geq \mu_1\lb P\rb^k >0,
\]
and therefore $P^{\otimes k}\in \inter\lb\mathcal{T}_k\rb$.
\end{proof}

Theorem \ref{thm:FullyGenkEA} has some important consequences. To illustrate these, consider first the case where
\[
\mathcal{T}_k = \mathcal{P}\lb \gC^{\otimes_{\max} k}_1,\gC^{\otimes_{\min} k}_2\rb,
\]
and note that $P^{\otimes k}\in \mathcal{T}_k$ if and only if $P$ is $k$-entanglement annihilating. Note that $\mathcal{T}_1$ is just the set of $(\gC_1,\gC_2)$-positive maps. If neither $\gC_1$ nor $\gC_2$ are classical, then we have $\mathcal{EB}\lb \gC_1, \gC_2\rb\subsetneq\mathcal{T}_1$ and there exists a linear map $P:V_1\ra V_2$ satisfying
\[
P\in \bounda\lb \mathcal{EB}\lb \gC_1, \gC_2\rb\rb\cap \inter\lb \mathcal{T}_1\rb .
\] 
By Theorem \ref{thm:FullyGenkEA}, we have $P^{\otimes k}\in \inter\lb \mathcal{T}_k\rb$ for every $k\in\N$. Moreover, there exists a linear map $R:V_1\ra V_2$ such that for every $\varepsilon>0$ the map $P_\varepsilon = P + \varepsilon R\notin \mathcal{EB}\lb \gC_1, \gC_2\rb$. Combining these two facts shows, that for every $k\in \N$ there exists an $\varepsilon>0$ such that $P^{\otimes k}_\varepsilon\in \mathcal{T}_k$, but $P_\varepsilon$ is not entanglement breaking. We have shown the following theorem:

\begin{thm}\label{thm:ConcretekEA1}
Let $\gC_1\subset V_1$ and $\gC_2\subset V_2$ be non-classical cones. For any $k\in\N$, there exists a linear map $P:V_1\ra V_2$ such that: 
\begin{enumerate}
\item We have $P^{\otimes k}\lb \gC^{\otimes_{\max} k}_1\rb\subseteq \gC^{\otimes_{\min} k}_2$.
\item The map $P$ is not $(\gC_1,\gC_2)$-entanglement breaking.
\end{enumerate}
\end{thm}

In the case of $\gC_1=\PSD\lb \C^{d_1}\rb$ and $\gC_2=\PSD\lb \C^{d_2}\rb$ we point out the following consequences of Theorem \ref{thm:FullyGenkEA} some of which have been appeared in the literature before, but others are new. These consequences all exploit the linear map $P:\M_{d}\ra \M_{d}$ given by 
\begin{equation}\label{equ:ReductionTypeMap}
P(X)=\Tr\lbr X\rbr\one_d - X/d ,
\end{equation}
which is entanglement breaking and at the boundary of completely positive maps.
\begin{itemize}
\item For $d,n,k\in \N$ consider $\mathcal{T}_k = \mathcal{P}_n(\PSD\lb (\C^{d})^{\otimes k}\rb,\PSD\lb (\C^{d})^{\otimes k}\rb)$, i.e., the mapping cones of $n$-positive maps. Let $P:\M_d\ra \M_d$ denote the map from \eqref{equ:ReductionTypeMap} and note that $P\in \inter\lb \mathcal{T}_1\rb$ provided that $n<d$. We can now find a linear map $S:\M_d\ra \M_d$ such that $P_\varepsilon = P-\varepsilon S$ is not completely positive for any $\varepsilon>0$. For any $k\in \N$ we can now apply Theorem \ref{thm:FullyGenkEA} to find $\varepsilon>0$ such that $P^{\otimes k}_\varepsilon$ is $n$-positive, but not completely positive. For $n=2$, this recovers a result from~\cite{watrous2004many}. 
\item A slight modification of the previous argument considers the map $Q = P\otimes \vartheta_d\circ P$, where $P$ is the map from \eqref{equ:ReductionTypeMap}. It is easy to see that $Q$ is at the boundary of both the completely positive and the completely copositive maps, but in the interior of positive maps (this follows from Lemma \ref{lem:Intermediate}). An argument similar to the one above shows for any $k\in \N$ that there are linear maps $Q_\varepsilon:\M_d\ra \M_d$ that are neither completely positive nor completely copositive, but such that $Q_\varepsilon^{\otimes k}$ is positive. This recovers a result from~\cite{muller2016positivity}.  
\item For $d,k\in \N$ consider $\mathcal{T}_k = \text{Dec}(\PSD\lb (\C^{d})^{\otimes k}\rb,\PSD\lb (\C^{d})^{\otimes k}\rb)$ the mapping cones of decomposable maps, i.e., linear maps that can be written as a sum of a completely positive and a completely copositive map. The map $Q = P\otimes \vartheta_d\circ P$ is also in the interior of decomposable maps (again this is shown by Lemma \ref{lem:Intermediate}). An argument similar to the ones above shows for any $k\in\N$ that there exists a linear map $Q_\varepsilon:\M_d\ra \M_d$ that is neither completely positive nor completely copositive, but such that $Q^{\otimes k}_\varepsilon$ is decomposable. This answers a question left open in recent work~\cite{muller2018decomposability,muller2021decomposable}.  
\end{itemize}

\section{Proof of Proposition \ref{prop:StrongPPT2}}
\label{app:ProofOfProp}

Recall the following class of maps introduced in~\cite{christandl2019composed}.

\begin{defn}[\cite{christandl2019composed}] \label{def:kEB}
For $k\in\N$ a linear map $T:\M_{m}\ra \M_{n}$ is called \emph{$k$-entanglement breaking} if 
\[
(\ident_k\otimes T)\lb \PSD\lb \C^k\otimes \C^{m}\rb\rb\subseteq \PSD\lb \C^k\rb\otimes_{\min} \PSD\lb \C^{n}\rb.
\]
\end{defn}
 
Here, we will focus on the $2$-entanglement breaking maps $T:\M_3\ra \M_3$, which has a very special structure compared to the sets of general $k$-entanglement breaking maps. The following characterization was obtained in~\cite{christandl2019composed}:

\begin{thm}[{\cite[Theorem 3.1]{christandl2019composed}}] \label{thm:charact2EB3}
A linear map $T:\M_3\ra \M_3$ is $2$-entanglement breaking if and only if both $T$ and 
$\vartheta_3\circ T$ are $2$-positive.
\end{thm} 

Furthermore, we use the following theorem from~\cite{christandl2019composed}:

\begin{thm}[{\cite[Theorem 2.1]{christandl2019composed}}] \label{thm:compose2EB3}
If $T_1 : \M_3 \to \M_3$ and $T_2 : \M_3 \to \M_3$ are $2$-entanglement breaking, then the composition $T_1 \circ T_2$ is entanglement breaking. 
\end{thm}

The proof of Proposition \ref{prop:StrongPPT2} is very easy:

\begin{proof}[Proof of Proposition \ref{prop:StrongPPT2}]
By Theorem \ref{thm:charact2EB3}, the maps $S$ and $T$ are $2$-entanglement breaking. By the Choi--Jamiolkowski isomorphism, the conclusion of Proposition \ref{prop:StrongPPT2} is equivalent to saying that $T\circ P\circ S$ is entanglement breaking for every positive map $P:\M_3\ra \M_3$. Since $T$ and $P \circ S$ are $2$-entanglement breaking, this follows from Theorem~\ref{thm:compose2EB3}.
\end{proof}

We want to close this appendix with two side remarks regarding the special structure of $2$-entanglement breaking maps $T:\M_3\ra \M_3$. First, we note the following corollary of Theorem \ref{thm:charact2EB3}:

\begin{cor}\label{cor:2EBbyAdjoints}
A linear map $T:\M_3\ra \M_3$ is $2$-entanglement breaking if and only if its adjoint $T^*$ is $2$-entanglement breaking.
\end{cor} 

Corollary \ref{cor:2EBbyAdjoints} turns out to be very special, since adjoints of $2$-entanglement breaking maps usually fail being $2$-entanglement breaking themselves. In fact, the set of $2$-entanglement breaking maps $T:\M_d\ra \M_d$ is not closed under adjoints for any $d\geq 4$ (as shown in~\cite{devendra2021mapping}). 

In the proof of Proposition \ref{prop:StrongPPT2}, we used the fact that the composition $P \circ S$ of a positive map and a $2$-entanglement breaking map is $2$-entanglement breaking as well. Recall that for any entanglement breaking map $T:\M_{d_2}\ra\M_{d_3}$ and any positive map $P:\M_{d_1}\ra \M_{d_2}$ the composition $T\circ P$ is also entanglement breaking. It seems unlikely that the analogous property holds for general $2$-entanglement breaking maps, but in the case of $d=3$ we have the following:

\begin{lem}
Let $T:\M_{3}\ra \M_{3}$ be a $2$-entanglement breaking map $P:\M_3\ra \M_3$ be a positive map. Then $T\circ P$ is $2$-entanglement breaking.
\label{lem:EasyLem}
\end{lem}
\begin{proof}
By Corollary \ref{cor:2EBbyAdjoints}, the composition $T\circ P$ is $2$-entanglement breaking if and only if its adjoint $P^*\circ T^*$ is $2$-entanglement breaking. Since $T^*$ is $2$-entanglement breaking by Corollary \ref{cor:2EBbyAdjoints}, the result follows.
\end{proof}

\section{Proof of Lemma \ref{lem:TPBanachVsCone}}
\label{sec:ProofLemTPBanachVsCone}

In the following, let $X$ denote a finite-dimensional normed space. Recall the subspace $X_k \subset \lb \R\oplus X\rb^{\otimes k}$ given by
\[
X_k := \text{span}\lb\lset e^{\otimes k}_0\rset\cup X^{\otimes k}\rb\subset (\R\oplus X)^{\otimes k} ,
\]
and, as before, we denote by ${\Pi}_{X_k}:\lb \R\oplus X\rb^{\otimes k} \ra \lb \R\oplus X\rb^{\otimes k}$ the orthogonal projection onto $X_k$. Consider the projection $S:(\R\oplus X)\otimes (\R\oplus X)\ra (\R\oplus X)\otimes (\R\oplus X)$ given by
\[
S = \frac{1}{2}(\ident_V\otimes \ident_V + A \otimes A),
\]
where $A:(\R\oplus V)\ra (\R\oplus V)$ is the reflection $(t,x) \mapsto (t,-x)$. We will need the following lemma:

\begin{lem}
We have
\begin{equation}
\Pi_{X_k} = \Pi_{i,j\in \lset 1,\ldots ,k\rset, i\neq j} S_{i,j},
\label{equ:ProjByS}
\end{equation}
where $S_{i,j}:X^{\otimes k}\ra X^{\otimes k}$ applies $S$ to the pair of tensor factors labeled $(i,j)$. Moreover, we have 
\[
\Pi_{X_k}\lb \gC^{\otimes_{\max} k}_X\rb\subset \gC^{\otimes_{\max} k}_X \quad\text{ and }\quad\Pi_{X_k}\lb \gC^{\otimes_{\min} k}_X\rb\subset \gC^{\otimes_{\min} k}_X .
\]
\label{lem:SeparableProj}
\end{lem}

\begin{proof}
It is easy to verify that 
\[
S(e_i\otimes e_j)=\begin{cases} e_i\otimes e_j, &\text{ if } i=j=0\text{ or }i,j\geq 1 \\ 0, &\text{ otherwise.}\end{cases}
\]
We conclude that \eqref{equ:ProjByS} holds. For the first part of the second claim, note that for any $\tilde{\phi}_1,\tilde{\phi}_2\in \gC^*_X$ of the form $\tilde{\phi}_i=e^*_0 + \phi_i$ for $\phi_i\in B_{X^*}$ and any $z\in \gC^{\otimes_{\max}2}_X$ we have 
\[
(\phi\otimes\psi)S(z) = \frac{1}{2}\lbr(\tilde{\phi}_1\otimes\tilde{\phi}_2)(z) + (\tilde{\phi}_1'\otimes\tilde{\phi}_2')(z)\rbr\geq 0,
\]
where $\tilde{\phi}'_i=e^*_0 - \phi_i\in \gC^*_X$. The general case follows in the same way. For the second part of the second claim, note that 
\[
S(x_1\otimes x_2) = \frac{1}{2}\lbr x_1\otimes x_2 + x_1'\otimes x_2'\rbr\in \gC^{\otimes_{\min}}_X ,
\]
where $x_i=I(x_i)\in \gC_X$. The general case follows in the same way.
\end{proof}

\begin{proof}[Proof of Lemma \ref{lem:TPBanachVsCone}]
By Lemma \ref{lem:SeparableProj} we have 
\[
{\Pi}_{X_k}(\gC^{\otimes_{\max} k}_X) \subseteq \gC^{\otimes_{\max} k}_X\cap X_k.
\]
Conversely, we have
\[
\gC^{\otimes_{\max} k}_X\cap X_k = {\Pi}_{X_k}(\gC^{\otimes_{\max} k}_X\cap X_k ) \subseteq {\Pi}_{X_k}(\gC^{\otimes_{\max} k}_X).
\]
This shows that $\gC^{\otimes_{\max} k}_X\cap X_k = {\Pi}_{X_k}(\gC^{\otimes_{\max} k}_X)$. Now, note that $\gC^*_X=\gC_{X^*}$ for any Banach space $X$. Therefore, $\tilde{\phi}\in \gC^*_X$ if and only if $\tilde{\phi} = t e^*_0 + \phi$ for some functional $\phi\in X^*$ satisfying $\|\phi\|_*\leq t$. Given $\tilde{z}\in X_k$ of the form $\tilde{z}=e^{\otimes k}_0 + z$ and $\tilde{\phi}_1,\ldots ,\tilde{\phi}_k\in \gC^*_X$, we have
\[
(\tilde{\phi}_1\otimes \cdots \otimes \tilde{\phi}_k)(\tilde{z}) = s_1\cdots s_k + (\phi_1\otimes\cdots \otimes\phi_k)(z),
\]
where $s_i\geq \|\phi_i\|_*$ for each $i\in\lset 1,\ldots ,k\rset$. If $\tilde{z}\in \gC_{X^{\otimes_{\varepsilon} k}}$, we conclude that $(\tilde{\phi}_1\otimes\cdots \otimes\tilde{\phi}_k)(\tilde{z})\geq 0$ whenever $\tilde{\phi}_1,\ldots ,\tilde{\phi}_k\in \gC^*_X$, and hence $\tilde{z}\in \gC_X^{\otimes_{\max}k}$. On the other hand, if $\tilde{z}=e^{\otimes k}_0 + z\in \gC_X^{\otimes_{\max}k}$, then $1 \pm (\phi_1\otimes\cdots \otimes\phi_k)(z)\geq 0$ for all functionals $\phi_1,\ldots ,\phi_k\in B_{X^*}$. Therefore, we conclude that $\tilde{z}\in \gC_{X^{\otimes_{\varepsilon}k}}$. 

By Lemma \ref{lem:SeparableProj} we have 
\[
{\Pi}_{X_k}(\gC^{\otimes_{\min} k}_X) \subseteq \gC^{\otimes_{\min} k}_X\cap V_k.
\]
Conversely, we have
\[
\gC^{\otimes_{\min} k}_X\cap X_k = {\Pi}_{X_k}(\gC^{\otimes_{\min} k}_X\cap X_k ) \subseteq {\Pi}_{X_k}(\gC^{\otimes_{\min} k}_X).
\]
This shows that $\gC^{\otimes_{\min} k}_X\cap X_k = {\Pi}_{X_k}(\gC^{\otimes_{\min} k}_X)$. For $\tilde{x}^{(1)},\ldots ,\tilde{x}^{(k)}\in \gC_X$ written as $\tilde{x}^{(i)}=e_0 + x^{(i)}$ for each $i\in\lset 1,\ldots ,k\rset$, we have ${\Pi}_{X_k}\lb \tilde{x}^{(1)}\otimes \cdots \otimes\tilde{x}^{(k)}\rb=e_0^{\otimes k} + x^{(1)}\otimes\cdots \otimes x^{(k)}$ and clearly 
\[
\|x^{(1)}\otimes\cdots \otimes x^{(k)}\|_{\pi_k(X)} \leq  \|x^{(1)}\|\cdots \|x^{(k)}\|\leq 1.
\] 
Therefore, ${\Pi}_{X_k}\lb \tilde{x}^{(1)}\otimes \cdots \otimes\tilde{x}^{(k)}\rb\in \gC_{X^{\otimes_\pi k}}$ and by definition of $\otimes_{\min}$ we have ${\Pi}_{X_k}\lb \gC_X^{\otimes_{\min}k}\rb \subseteq \gC_{X^{\otimes_\pi k}}$. For the converse inclusion consider $\tilde{z}\in \gC_{X^{\otimes_\pi k}}$ of the form $\tilde{z}= \| z\|_{\pi_k(X)} e^{\otimes k}_0 + z$. For some $n\in\N$ there exist $x^{(1)}_i,\ldots ,x^{(k)}_i\in X$ for $i\in\lset 1,\ldots ,n\rset$ such that $z=\sum_{i} x^{(1)}_i\otimes \cdots x^{(k)}_i$ and $\sum_{i}\| x^{(1)}_i\|\cdots \|x^{(k)}_i\|=\|z\|_{\pi_k}\leq 1$. For each $i$ and $j$ we define $\tilde{x}^{(j)}_i = \|x^{(j)}_i\|e_0 + x^{(j)}_i\in \gC_X$ such that 
\[
\tilde{z} = \sum_i \|x^{(1)}_i\|\cdots \|x^{(k)}_i\|e^{\otimes k}_0 + \sum_{i}x^{(1)}_i\otimes\cdots \otimes x^{(k)}_i = {\Pi}_{X_k}\lb \sum_i \tilde{x}^{(1)}_i\otimes\cdots \otimes \tilde{x}^{(k)}_i \rb.
\]
By convexity and since $e^{\otimes k}_0\in {\Pi}_{X_k}(\gC^{\otimes_{\min} k}_X)$ we find $\gC_{X^{\otimes_\pi k}}\subseteq {\Pi}_{X_k}(\gC^{\otimes_{\min} k}_X)$.
\end{proof}

\section{Two results for symmetric cones}
\label{sec:Jordan}

A symmetric cone $\gC\subset V$ is a proper cone in a Euclidean vector space $V$ that it selfdual, i.e., it satisfies $\gC=\gC^*$, and such that the automorphism group $\Aut(\gC)$ of $\gC$ acts transitively on its interior $\text{int}\lb \gC\rb$. It is well-known that symmetric cones are closely related to Jordan algebras. Let $G$ denote the connected component of $\Aut\lb C\rb$ containing the identity in the orthogonal group $O(V)$, and let $K=G\cap O(V)$. By \cite[Proposition I.1.9]{faraut1994analysis} we may choose an element $e\in C$ such that $K$ arises as the stabilizer of $e$ in $\Aut\lb C\rb$. By \cite[Theorem III.3.1]{faraut1994analysis} the vector space $V$ can then be equipped with a product turning it into a Euclidean Jordan algebra with identity element $e$ such that 
\begin{equation}
\gC = \lset x^2 ~:~x\in V\rset.
\label{equ:ConeAsSquares}
\end{equation}
From this description and the classification theorem of Euclidean Jordan algebras, it is not surprising that being symmetric is a restrictive property. Indeed, the classification theorem due to Vinberg~\cite{vinberg1963theory} (see also~\cite{faraut1994analysis}) shows that every indecomposable symmetric cone is isomorphic to one of the following examples:
\begin{itemize}
\item The positive semidefinite cones $\PSD(\R^d)$ over the real numbers.
\item The positive semidefinite cones $\PSD(\C^d)$ over the complex numbers.
\item The positive semidefinite cones $\PSD(\qua^d)$ over the quaternions.
\item The $3\times 3$ positive semidefinite cone $\PSD(\oc^3)$ over the octonions.
\item The Lorentz cones $\gL_n$.
\end{itemize}
Here, we will prove two results for positive maps on symmetric cones: The first result generalizes the Sinkhorn normal form of positive maps between cones of positive-semidefinite matrices to positive maps between symmetric cones. The second result generalizes a result on inclusion-constants of matrix convex sets~\cite[Theorem 6.6]{passer2018minimal} to the setting of symmetric cones. To make our presentation self-contained we will review well-known constructions from the theory of Euclidean Jordan algebras and we refer the reader to the book~\cite{faraut1994analysis} for more details.

\subsection{Sinkhorn-type scaling on symmetric cones}

Let $\gC\subset V$ denote a symmetric cone in a Euclidean Jordan algebra $V$ with identity element $e\in \gC$ such that \eqref{equ:ConeAsSquares} holds. For any $x\in V$ we consider 
\[
m(x) = \min\lset k>0 ~:~\lset e,x,x^2,\ldots ,x^k\rset \text{ linearly independent} \rset,
\]
and we set $d=\max\lset m(x)~:~x\in V\rset$. An element $x\in V$ is called \emph{regular} if $m(x)=d$. By~\cite[Proposition II.2.1]{faraut1994analysis} the regular elements form an open and dense subset in $V$. Moreover, for each $k\in\lset 1,\ldots ,d\rset$ there exists a homogeneous polynomial $a_k$ of degree $k$ such that 
\[
x^d - a_1(x)x^{d-1} + a_2(x)x^{d-2} + \cdots + (-1)^d a_d(x)e=0
\]  
for every regular $x\in V$. By continuity the previous equation extends to the whole Jordan algebra $V$, and we set $\det(x):=a_d(x)$. An element $x\in V$ is called invertible if $\det(x)\neq 0$ and we denote by $V_I\subset V$ the set of invertible elements in $V$. The inverse $\inv:V_I\ra V_I$ is given by 
\[
\inv(x) = \frac{1}{\det(x)}\lb x^{d-1}-a_1(x)x^{d-2} +\cdots + (-1)^{d-1}a_{d-1}(x)e\rb,
\]
and sometimes we write $x^{-1}$ instead of $\inv(x)$. It turns out that the interior $\gC^\circ$ arises as
\[
\gC^\circ = \lset x^2~:~x\in V_I\rset ,
\]
and consequently any element of $\gC^\circ$ is invertible.

For each $x\in V$ there is a left multiplication $L_x:V\ra V$ given by $L_x y = xy$ for $y\in V$. Using this operator we define the quadratic representation $Q_x:V\ra V$ as 
\[
Q_x = 2L_x^2 - L_{x^2}.
\]
We have the following (see~\cite[Section II.3.~and Proposition II.4.4]{faraut1994analysis}).

\begin{lem}[Properties of the quadratic representation~\cite{faraut1994analysis}]
For each $x\in V$ we have 
\begin{enumerate}
\item $Q_x$ is a self-adjoint operator on $V$ and $Q_x\in \Aut(\gC)$.
\item $Q_x(e) = x^2$.
\item $Q_x^{-1} = Q_{x^{-1}}$ whenever $x$ is invertible.
\end{enumerate}
\label{lem:PropsOfQuadr}
\end{lem}

With the terminology introduced before, we can now state and prove the main result of this section. Our proof follows the lines (and generalizes) a proof for positive maps between the cones of positive semidefinite matrices with complex entries~\cite[Lemma 1.14]{idel2013structure}. See also~\cite{idel2016review} for a comprehensive review of similar results:

\begin{thm}[Sinkhorn-type scaling]\label{thm:SinkhornTypeScaling}
Consider symmetric cones $\gC_1\subset V_1$ and $\gC_2\subset V_2$ in Euclidean vector spaces $V_1$ and $V_2$. For each $i\in\lset 1,2\rset$ we denote by $G_i$ the identity component in the automorphism group $\Aut\lb \gC_i\rb$ and by $e_i\in \gC_i$ any element with stabilizer $G_i\cap O(V_i)$,  such that $\|e_1\|=\|e_2\|$ (where $\|\cdot\|$ denotes the Euclidean norm on $V_1$ or $V_2$). For any linear map $P:V_1\ra V_2$ satisfying
\[
P(\gC_1)\subseteq \gC^\circ_2 ,
\]
there are automorphisms $A\in \Aut\lb \gC_1\rb$ and $B\in \Aut\lb \gC_2\rb$ such that the linear map 
\[
\tilde{P}=B\circ P\circ A
\]
satisfies $\tilde{P}(e_1)=e_2$ and $\tilde{P}^*(e_2)=e_1$. 
\end{thm}

Theorem \ref{thm:SinkhornTypeScaling} includes the following special cases.
\begin{enumerate}
\item Let $A$ be a $n \times n$ matrix with positive entries. Then there exist diagonal matrices $D_1$, $D_2$ with positive diagonal elements such that the matrix $D_1 A D_2$ is bistochastic (i.e., the sum of elements in each row and each column is~$1$). This is known as Sinkhorn's theorem \cite{Sinkhorn64} and can be deduced by applying Theorem \ref{thm:SinkhornTypeScaling} with $\gC_1$ and $\gC_2$ being the symmetric cone $\R_+^n$.
\item For $n \geq 1$, the Lorentz cone
\[ L_n = \left\{ (x_0,\dots,x_{n-1}) \in \R^n \, : \, x_0 \geq \sqrt{x_1^2+ \cdots +x_{n-1}^2} \right\} \]
is a symmetric cone. By Theorem~\ref{thm:SinkhornTypeScaling}, given a linear map $P : \R^n \to \R^m$ (identified with a matrix) such that $P(L_n) \subset L_m^\circ$, there exist $A_1 \in \Aut(L_n)$, $A_2 \in \Aut(L_m)$ such that $A_2 P A_1$ is the block-diagonal matrix $\begin{pmatrix} \lambda & 0 \\ 0 & M\end{pmatrix}$ with $\lambda > 0$ and $M \in \R^{(n-1)\times (m-1)}$. Moreover, using the singular value decomposition, the matrix $M$ can be assumed to be diagonal with nonnegative coefficients, recovering \cite[Theorem 3.4]{hildebrand2007lmi}.
\item For $d\in\N$, the positive semidefinite cone $\PSD(\C^d)$ is a symmetric cone. Given a positive map $P:\M_{d_1}\ra \M_{d_2}$ satisfying
\[
P(\PSD(\C^{d_1}))\subset \PSD(\C^{d_2})^\circ,
\]
we can apply Theorem \ref{thm:SinkhornTypeScaling} to find automorphisms $A \in \Aut(P(\PSD(\C^{d_1})))$, $B \in \Aut(P(\PSD(\C^{d_2})))$ such that $B\circ P\circ A$ is unital and trace-preserving. This recovers~\cite[Theorem 4.7]{gurvits2004classical} and the aforementioned normal form from~\cite[Lemma 1.14]{idel2013structure}.
\end{enumerate}

\begin{proof}

Under the stated assumptions, the Euclidean vector spaces $V_1$ and $V_2$ can be equipped with products turning them into Euclidean Jordan algebras with identity elements $e_1$ and $e_2$ respectively. The symmetric cones $\gC_1$ and $\gC_2$ then arise as cones of squares as in \eqref{equ:ConeAsSquares}. We define the slice
\[
\gC^s_1 = \lset x\in \gC_1 ~:~\braket{e_1}{x}=1\rset,
\]
and we note that $\gC^s_1$ is compact and convex. Next, we define $M:\gC^s_1\ra \gC^s_1$ by
\[
M(x) = \frac{\lb \inv_1\circ P^*\circ \inv_2\circ P\rb (x)}{\braket{e_1}{\lb \inv_1\circ P^*\circ \inv_2\circ P\rb(x)}},
\] 
where $\inv_1$ and $\inv_2$ denote the inverses on the Jordan algebras $V_1$ and $V_2$ respectively as defined above. The map $M$ is well-defined since $P(x)\in \gC^\circ_2$ and $P^*(y)\in \gC^\circ_1$ are invertible for every $x\in \gC_1$ and $y\in \gC_2$. Moreover, the map $M$ is continuous as a composition of continuous maps. By Brouwer's fixed point theorem there exists $x\in \gC^s_1$ such that $M(x)=x$. This implies that $P^* \circ \inv_2 \circ P = \lambda x^{-1}$ for some $\lambda > 0$. Set $y=P(x)$. Let $\sqrt{x}\in V_1$ and $\sqrt{y}\in V_2$ denote square roots of $x$ and $y$ with respect to the respective Jordan algebra structures on $V_1$ and $V_2$, i.e., elements $\sqrt{x}\in V_1$ and $\sqrt{y}\in V_2$ satisfying $x=(\sqrt{x})^2$ and $y=(\sqrt{y})^2$. Next, we introduce the automorphisms $A = Q_{\sqrt{x}}$ and $B=Q^{-1}_{\sqrt{y}}$ as quadratic representations. Defining $\tilde{P}=B\circ P \circ A$ we can verify that
\[
\tilde{P}(e_1) = Q^{-1}_{\sqrt{y}}\circ P\circ Q_{\sqrt{x}}(e_1) = Q^{-1}_{\sqrt{P(x)}}\lb P(x)\rb = e_2,
\]  
using the properties from Lemma \ref{lem:PropsOfQuadr}. Since $A$ and $B$ are self-adjoint we compute 
\[
\tilde{P}^*(e_2) = Q_{\sqrt{x}}\circ P^*\circ Q^{-1}_{\sqrt{y}}(e_2) = Q_{\sqrt{x}}\circ P^*\circ \inv_2\circ P(x) = Q_{\sqrt{x}}\lb \lambda x^{-1}\rb = \lambda e_1 ,
\]
where we used the properties from Lemma \ref{lem:PropsOfQuadr}. Finally, 
\[ \|e_2\|^2 = \langle e_2,\tilde{P} (e_1) \rangle = \langle \tilde{P}^* (e_2), e_1 \rangle = \lambda \|e_1\|^2,\]
and since $\|e_1\|=\|e_2\|$ we conclude that $\lambda =1$, finishing the proof.
\end{proof}

A direct consequence of the previous theorem is the following corollary: 

\begin{cor}
Consider symmetric cones $\gC_1\subset V_1$ and $\gC_2\subset V_2$ in Euclidean vector spaces $V_1$ and $V_2$. For each $i\in\lset 1,2\rset$ we denote by $G_i$ the identity component in the automorphism group $\Aut\lb \gC_i\rb$ and by $e_i\in \gC_i$ any element with stabilizer $G_i\cap O(V_i)$, such that $\|e_1\|=\|e_2\|$ (where $\|\cdot\|$ denotes the Euclidean norm on $V_1$ or $V_2$). The following are equivalent:
\begin{enumerate}
\item The pair $(\gC_1,\gC_2)$ is resilient.
\item Every $(\gC_1,\gC_2)$-entanglement annihilating map $P:V_1\ra V_2$ satisfying $P(e_1)=e_2$ and $P^*(e_2)=e_1$ is $(\gC_1,\gC_2)$-entanglement breaking. 
\end{enumerate}
\end{cor}

\begin{proof}
It is clear that the first statement implies the second. For the converse direction assume that $(\gC_1,\gC_2)$ is not resilient and let $R:V_1\ra V_2$ be an entanglement annihilating map that is not entanglement breaking. Since the set of entanglement breaking maps is closed, there exists an $\varepsilon>0$ such that $R_\varepsilon:V_1\ra V_2$ given by $R_\varepsilon = R + \varepsilon e_2\braket{e_1}{\cdot}$ is entanglement annihilating and not entanglement breaking. Using Theorem \ref{thm:SinkhornTypeScaling} we find automorphisms $A\in \Aut\lb \gC_1\rb$ and $B\in \Aut\lb \gC_2\rb$ such that the map $P=B\circ R_\varepsilon \circ A$ satisfies $P(e_1)=e_2$ and $P^*(e_2)=e_1$. Moreover, it is easy to see that $P$ is entanglement annihilating and not entanglement breaking. This finishes the proof.
\end{proof}

\subsection{Breaking entanglement of a symmetric cone and $\gC_{\ell^k_1}$}

Let $\gC \subset V$ be a symmetric cone. Equip $V$ with the associated Jordan algebra structure. In this section, we will often use the \emph{spectral theorem} on Euclidean Jordan algebras~\cite[Theorem III.1.2]{faraut1994analysis}: For any $x\in V$ there exists a Jordan frame $c_1,c_2,\ldots ,c_k$, i.e., a complete system of orthogonal, primitive idempotents, and unique $\lambda_1,\ldots ,\lambda_k\in\R$ such that 
\begin{equation}
x = \sum^k_{i=1} \lambda_i c_i .
\label{equ:spectrJordan}
\end{equation}
Using self-duality of the cone $\gC$ and the properties of Jordan frames, it it easy to show that $x\in\gC$ if and only if $\lambda_i\geq 0$ for every $i\in\lset 1,\ldots ,k\rset$. Given a spectral decomposition \eqref{equ:spectrJordan} of $x\in V$, we define $x_+ := \sum_{i~:~\lambda_i\geq 0} \lambda_i c_i$ and $x_- := \sum_{i~:~\lambda_i< 0} |\lambda_i| c_i$. Clearly, we have $x_+,x_-\in \gC$, $x_+x_-=x_-x_+=0$ and $x=x_+-x_-$. Finally, we define $|x| := x_++x_-$ and note that $|x|\in \gC$. 

We start with the following lemma:

\begin{lem}
Consider the symmetric cone
\[
\gC = \lset x^2 ~:~x\in V\rset.
\]
for a Euclidean Jordan algebra $V$ with identity element $e\in V$. Then, we have the following:
\begin{enumerate}
\item If $e+x\in \gC$ and $e-x\in \gC$ for some $x\in V$, then $e-x^2\in \gC$. 
\item If $e-x^2\in \gC$ for some $x\in V$, then $e-x\in \gC$.
\end{enumerate}
\label{lem:Basic}
\end{lem}

\begin{proof}
Given $x\in V$ consider the spectral decomposition $x=\sum^k_{i=1} \lambda_i c_i$ for a Jordan frame $c_1,c_2,\ldots ,c_k\in V$. If $e+x\in \gC$ and $e-x\in \gC$, then we have $1\pm \lambda_i\geq 0$ and therefore $1-\lambda^2_i\geq 0$ for all $i\in\lset 1,\ldots ,k\rset$. We conclude that 
\[
e-x^2 = (e-x)(e+x) = \sum^k_{i=1} (1-\lambda^2_i) c_i \in \gC,
\]
since $c_i=c^2_i\in \gC$ for every $i\in\lset 1,\ldots ,k\rset$. This shows the first statement.

If $e-x^2\in \gC$, then we have $\lambda^2_i\leq 1$ for all $i\in\lset 1,\ldots ,k\rset$. We conclude that 
\[
e-x = \sum^k_{i=1} (1-\lambda_i) c_i \in \gC,
\] 
since $\lambda_i\leq 1$ and $c_i\in \gC$ for any $i\in\lset 1,\ldots ,k\rset$.
\end{proof}

We will now prove a lemma identifying a useful property of certain elements in $\gC \otimes_{\max} \gC_{\ell^k_1}$. For this and in the following, we will identify elements in $V\otimes \R^{k+1}$ with vectors $(x_0,x_1,\ldots ,x_k)$ for $x_0,\ldots ,x_k\in V$.   

\begin{lem}
Consider the symmetric cone
\[
\gC = \lset x^2 ~:~x\in V\rset.
\]
for a Euclidean Jordan algebra $V$ with identity element $e\in V$. If $(e,x_1,\dots,x_k) \in \gC\otimes_{\max}\gC_{\ell^k_1}$, then $\sqrt{k}e - \sum^k_{i=1}|x_i|\in \gC$.
\label{lem:intermed}
\end{lem}

\begin{proof}
Note that $(e,x_1,\ldots ,x_k)\in \gC\otimes_{\max} \gC_{\ell^k_{1}}$ if and only if 
\[
e + \sum^k_{i=1} s_i x_i \in \gC \quad\text{ and }\quad e - \sum^k_{i=1} s_i x_i \in \gC ,
\]
for all $s\in\lset +1,-1\rset^k$. By Lemma \ref{lem:Basic} we conclude that 
\[
e - \sum^k_{i=1} x^2_i - \sum_{i\neq j} s_i s_j x_i x_j \in \gC.
\]
Averaging over all choices $s\in\lset +1,-1\rset^k$ shows that 
\[
e - \sum^k_{i=1} x^2_i \in \gC.
\] 
Now, we have
\[
\left( \sum_{i=1}^k |x_i| \right)^2 + \sum_{i<j} \left( |x_i|-|x_j| \right)^2 = k \sum_{i=1}^k |x_i|^2 = k \sum_{i=1}^k x_i^2,
\]
where we used that
\[
|x|^2=(x_++x_-)^2 = x^2_+ +x^2_-= (x_+-x_-)^2=x^2,
\] 
for any $x\in V$, since $x_+x_-=x_-x_+=0$. We conclude that 
\[
ke - \left( \sum_{i=1}^k |x_i| \right)^2 \in \gC,
\] 
implying that $\sqrt{k}e-\sum_{i=1}^k |x_i| \in \gC$ by Lemma \ref{lem:Basic}.
\end{proof}

Now we can prove the main result of this appendix:

\begin{proof}[Proof of Theorem \ref{thm:BreakingEntWithSymmetricCone}]
As explained in the introduction of this appendix, there exists a product turning $V$ into a Euclidean Jordan algebra with identity element $e\in V$ such that 
\[
\gC = \lset x^2 ~:~x\in V\rset.
\]
Consider $x=(x_0,x_1,\ldots ,x_k)\in \gC\otimes_{\max} \gC_{\ell^k_{1}}$ and set $x_\varepsilon = (x_0+\varepsilon e,x_1,\ldots ,x_k)$. By continuity, we could conclude that $(\ident_V\otimes I_{\sqrt{k}})(x)\in \gC \otimes_{\min} \gC_{\ell^k_1}$ if we could show that $(\ident_{V}\otimes I_{\sqrt{k}})(x_{\varepsilon})\in \gC\otimes_{\min}\gC_{\ell^k_1}$ for all $\varepsilon>0$. For any $\varepsilon>0$ we have $x_0+\varepsilon e\in\inter(\gC)$ and there exists an automorphism $\phi_\varepsilon\in \Aut(\gC)$ satisfying $\phi_\varepsilon(e)= x_0+\varepsilon e$. Therefore, we can write
\[
(\ident_{V}\otimes I_{\sqrt{k}})(x_{\varepsilon}) = (\phi_\varepsilon\otimes \ident_{\R^{k+1}})\circ (\ident_{V}\otimes I_{\sqrt{k}})(x')
\]
for any $\varepsilon>0$ and some $x'=(e,x'_1, \ldots x'_k)\in \gC \otimes_{\max} \gC_{\ell^k_{1}}$. Since $\phi_\varepsilon$ is positive for every $\varepsilon>0$ the proof would be finished by showing that
\[
(\ident_{V}\otimes I_{\sqrt{k}})(x) \in \gC\otimes_{\min} \gC_{\ell^k_1} ,
\]
for any $x=(e,x_1,\ldots ,x_k)\in \gC \otimes_{\max}\gC_{\ell^k_1}$. Let $x=(e,x_1,\ldots ,x_k)\in \gC \otimes_{\max} \gC_{\ell^k_1}$ and by Lemma \ref{lem:intermed} we conclude that 
\[
\sqrt{k}e-\sum_{i=1}^k |x_i| \in \gC .
\] 
Note that 
\[
(e_0+e_i)\otimes x_+ + (e_0-e_i)\otimes x_- = e_0 \otimes |x| + e_i\otimes x,
\]
for any $i\in\lset 1,\ldots , k\rset$. Since $e_0\pm e_i\in \gC_{\ell^k_1}$ for any $i\in\lset 1,\ldots , k\rset$ we have
\begin{align*}
(\sqrt{k}& e,x_1,\dots,x_k) \\
&= e_0 \otimes \left(\sqrt{k} e - \sum_{i=1}^k |x_i| \right) + \sum_{i=1}^k (e_0 + e_i) \otimes (x_i)_+ + \sum_{i=1}^k (e_0 -e_i) \otimes (x_i)_-\\
&\in \gC\otimes_{\min}\gC_{\ell^k_1}.
\end{align*}
This finishes the proof.
\end{proof}

\bibliographystyle{alpha}
\bibliography{Biblio.bib}

\end{document}